\documentclass[journal,onecolumn,final]{IEEEtran}
\usepackage{booktabs}

\usepackage{hyperref}
\hypersetup{colorlinks,
citecolor=black,%
filecolor=black,%
linkcolor=black,%
urlcolor=black}
\usepackage{bm,mathptmx,amsfonts,amssymb,amstext,textcomp,amsmath}
\usepackage{times,verbatim,multirow,epsfig,graphics,graphicx}
\usepackage[usenames]{color}
\usepackage{tikz}
\usepackage{lipsum}
\usepackage{karnaugh-map}

\usepackage[makeroom]{cancel}
\usepackage{algorithm}
\usepackage[noend]{algpseudocode}
\usepackage{amsopn}

\DeclareMathAlphabet\mathbfcal{OMS}{cmsy}{b}{n}

\newcommand{\nat}{\mathbb{N}}

\newcommand{\rea}{\mathbb{R}}

\newcommand{\E}{\mathbfcal{E}}
\newcommand{\G}{\mathbfcal{G}}

\newcommand{\V}{\mathbfcal{V}}

\newcommand{\bnot}{\sim}
\newcommand{\band}{\wedge}
\newcommand{\bor}{\vee}
\newcommand{\bimplies}{\rightarrow}
\newcommand{\biff}{\leftrightarrow}

\newtheorem{theorem}{\textbf{Theorem}}[section]

\newtheorem{propo}[theorem]{\textbf{Proposition}}
\newtheorem{definition}[theorem]{\textbf{Definition}}

\newcommand{\argmin}{\mathop{\operatorname{argmin}}}

\usepackage[final]{changes}

\usepackage{letltxmacro}
\makeatletter
\let\oldr@@t\r@@t
\def\r@@t#1#2{%
\setbox0=\hbox{$\oldr@@t#1{#2\,}$}\dimen0=\ht0
\advance\dimen0-0.2\ht0
\setbox2=\hbox{\vrule height\ht0 depth -\dimen0}%
{\box0\lower0.4pt\box2}}
\LetLtxMacro{\oldsqrt}{\sqrt}
\renewcommand*{\sqrt}[2][\ ]{\oldsqrt[#1]{#2} }

\begin{document}

\title{Multi-Agent Coordination of Thermostatically Controlled Loads by Smart Power Sockets for Electric Demand Side Management}

\author{Mauro~Franceschelli, Alessandro~Pilloni, Andrea Gasparri
\thanks{\scriptsize The research leading to these results has received funding from  the Italian Ministry of Research and Education (MIUR) with project ``CoNetDomeSys", code  RBSI14OF6H, under call SIR 2014.
}
\thanks{\scriptsize M.~Franceschelli is with the Department of Electrical and Electronic Engineering (DIEE), University of Cagliari, Cagliari, Italy. Email: mauro.franceschelli@diee.unica.it. Mauro Franceschelli is the corresponding author.

 A.~Pilloni is with the Department of Electrical and Electronic Engineering (DIEE), University of Cagliari, Cagliari, Italy. Email: alessandro.pilloni@diee.unica.it.

A.~Gasparri, is with the Department of Engineering, University of Roma Tre, Rome, Italy. Email: gasparri@dia.uniroma3.it.
}
}


%

\maketitle

\begin{abstract}
This paper presents a multi-agent control architecture and an online optimization method based on dynamic average consensus to coordinate the power consumption of a large population of Thermostatically Controlled Loads (TCLs). Our objective is to penalize peaks of power demand, smooth the load profile and enable Demand Side Management (DSM). The proposed architecture and methods exploit only local measurements of power consumption via Smart Power Sockets (SPSs) with no access to their internal temperature. No centralized aggregator of information is exploited and agents preserve their privacy by cooperating anonymously only through consensus-based distributed estimation, robust to node/link failure. The interactions among devices are designed to occur through an unstructured peer-to-peer (P2P) network over the internet. 

The architecture includes novel methods for parameter identification, state estimation and mixed logical modelling of TCLs and SPSs. It is designed from a multi-agent and plug-and-play perspective in which existing  household appliances can interact with each other in an urban environment.

Finally, a novel low cost testbed is proposed along with  numerical tests and an experimental validation.

\end{abstract}


\IEEEpeerreviewmaketitle

\section{Introduction} \label{sect1}

\IEEEPARstart{M}{atching} power generation and consumption is the fundamental problem of the power grid \cite{el2004optimal}, this issue is widely known to be made worse by volatile renewable power generation and hourly variations of urban power demand due to the time-correlation of the usage of domestic electric appliances for water heating, air conditioning, cooking etc. This problem may be ameliorated by increasing the flexibility of electric power demand to reduce short term urban electric load variations. To do so, a future vision of the power system that leverages Information and Communications Technology (ICT) to implement advanced control strategies to improve the flexibility and the efficiency of the DSM strategy is required ~\cite{deng2015survey,sabbah2014survey}.

About $40\%$ of total electric power demand in the US in 2017 was due to residential power consumption according to US Department of Energy \cite{DoE}. According to the same source, the largest share of residential power consumption is due to electric heating and cooling usually achieved by the so called Thermostatically Controlled Loads (TCLs) such as water heaters, freezers, boilers, radiators, and air conditioners.  This data varies greatly depending on the country, city and cost and availability of electricity with respect to natural gas or other fossil fuels. These appliances are controlled by a thermostat and therefore are characterized by a simple ON/OFF power consumption dynamics which can be modulated to act as energy storage devices to provide ancillary services to the smart grid, see \cite{Hao2015189,Tindemans15,Grammatico2015,franceschelli2016coordination,ecc18}. Thus, large populations of actively controlled TCLs can effectively add some flexibility to modulate the urban electric power demand.

Some promising methods for electric Demand Side Management (DSM) that focus on controlling TCLs can be found in \cite{Braslavsky13,Xing20145439,Grammatico2015}. There, centralized strategies and distributed decision making methods supported by a centralized information aggregator are exploited to address the problem.

In the electric demand side management community the keyword "distributed control" and "multi-agent" is associated to distributed decision making by smart homes or devices which interact with a centralized aggregator of information which collects data and updates reference set-points or price signals. In fact, real-time pricing strategies could be seen as a form of distributed decision making where a feedback loop on the system is closed by the energy provider that measures the power demand and broadcasts a unique time-varying price signal. This kind of control architecture may present issues with the privacy of the users or be vulnerable to DoS cyber-attacks on the information aggregator.

In this paper we propose a different framework where smart devices (agents), TCLs in particular, cooperate within a peer-to-peer network autonomously and anonymously with a small set of neighboring agents with the network graph. 
The proposed method aims to exploit only local and asynchronous anonymous interactions among the agents to optimize trough their emergent behavior a global objective function defined by their power consumption. In particular, we chose an objective function which incentivises the shaving of peak power consumption and reduction of electric load variations by the network of TCLs. Through the proposed framework, the power consumption of the network can be modulated, thus enabling the shaping of the electric load profile without any direct control action on any device or the sharing of power consumption information with a centralized coordinator.

The proposed method is paired with a multi-agent control architecture which is intended to exploit the cheapest possible hardware to enable cooperation among devices, i.e., a smart power socket, suitable to retrofit existing TCLs such as domestic water heaters, thus greatly reducing the cost of the infrastructure needed for the electric demand side management (DSM) program and therefore significantly reducing the cost of adoption by the users. This design choice makes the control problem more challenging because on each device the identification of the TCL dynamics, the estimation of its state on its control has to be carried out with only power consumption measurements and ON/OFF control capability.  For instance, in \cite{ParIdenti2012} the problem of identifying the dynamics of the TCL was addressed by either considering its internal temperature measurable or at least assuming that the temperature range of the thermostat (maximum and minimum temperature) was known in absolute terms. In this paper, we only exploit power consumption measurement to develop a real plug-and-play approach which does not require any system configuration by the user.

Furthermore, since we assume each agent in the network as being representative of a SPS connected to a TCL, we consider in fact agents modeled by a hybrid systems. The multi-agent coordination problem that we formulate in this paper consists in the online optimization (or cooperative model predictive control) of an objective function which is not separable and subject to non-convex local constraints on each agent due its hybrid dynamics. Nonetheless, despite the difficulty of the considered problem we provide a heuristic approach which converges up to a local minimum of the considered objective. Experimental observations and results indicate that the amount of modelling uncertainty and disturbances that affect each single system in the large scale network considered is significant and therefore there is little advantage on computing optimal control solutions (even if there were a method to compute them with our current working assumption) with respect to heuristic solutions.

\textbf{Literature review:}
In \cite{Braslavsky13} the authors propose a feedback control scheme for TCLs (Air conditioners), as opposed to open-loop and model-free strategies, where broadcasts of thermostat set-points offset changes to the ACs. This scheme requires readings from a  common power distribution connection point where total aggregate demand is measured, furthermore the authors assume to be able to change the reference temperature of local thermostat. Among many other differences with our work, we consider the local temperature reference of the TCL to be unknown and not accessible, in order to enable a cheap retrofit of existing devices by exploiting only SPSs.

In \cite{Xing20145439}, the authors propose a distributed algorithm to control a network of thermostatically controlled loads (TCLs) to match, in real time, the aggregated power consumption of a population of TCLs with the predicted power supply. The algorithm is a consensus-based optimization of an objective function which represents the sum of temperature differences with respect to a desired reference. In their work, the authors assume that each device runs an instance of a consensus algorithm to estimate the average power consumption in the network  at each iteration. Based on this information and on a centralized forecast of power consumption, the desired power consumption is assigned to each device. Among many other differences, in our work we deal with actual devices which are not supposed to be tailored for the DSM tasks, and therefore algorithms of the kind of \cite{Xing20145439} can not be implemented without ad-hoc hardware. Furthermore, we consider an online optimization process which is based on local asynchronous optimizations and is robust to changes in the network during execution. 

In \cite{Grammatico2015} an elegant approach based on mean-field control theory is proposed for the control of large populations of TCLs. The authors consider a quadratic cost objective  function of the temperature profiles of the TCLs to be optimized. The approach considers local decision making by the TCLs and a centralized aggregator of information that computes the average state values to be used as feedback and proves convergence of the strategy to a fixed point by exploiting the theory of contraction mappings and mean fields. In \cite{Grammatico20163315} the same authors generalize and extend the approach to consider also local convex constraints and several additional applications. In our scenario, apart from considering a different objective for the network of TCLs, the internal temperature of the TCL is not accessible and the TCLs are controlled by SPS, thus their dynamics is hybrid and is modeled by mixed integer local constraints. Furthermore we also deal with the issues of identification and estimation required to make our approach implementable in a real scenario with our proposed experimental testbed.   

In \cite{hadjicostis6426665} the problem of optimally dispatching a set of energy resources is considered. The authors formulate a convex optimization problem and propose  the so-called \emph{ratio-consensus} algorithm to enable the distributed decision making process to occur in an unbalanced directed graph.

In \cite{notarnicola2016duality,Notarnicola8472154} a distributed optimization method to solve min-max problems characterized by an objective function which is not separable and with local convex constraints is presented. The considered optimization problem is motivated by the peak-to-average ratio minimization problem in a network of TCLs. In oue setting, we consider a different objective function, non-convex local constraints due the modeling of the TCLs and SPSs and different control variables.



Finally, our method and architecture exploits a dynamic average consensus algorithm to enable the distributed estimation of the average planned power consumption of the whole network of TCLs. The dynamic consensus problem is an agreement problem, or consensus problem, where the agents aim to agree on the current average value of their inputs or reference signals, as opposed to their initial state value as commonly considered in the literature. The reader is referred to \cite{SolmazTutorial,spanos2005dynamic,Zhu2010322,Kia2015112,Cortes2015,FraGas2016,Franceschelli201969} for a comprehensive overview on the dynamic average consensus problem. 


Summarizing, the \textbf{main contributions} of this paper are:

\begin{itemize}
    \item A multi-agent DSM control architecture for the coordination of anonymous networks of  thermostatically controlled loads via smart power sockets;
    \item A method for power consumption model identification for TCLs based only on power consumption measurements;
    \item An hybrid observer for the estimation of the TCL internal state based only on power consumption measurements;
    \item A distributed and randomized online optimization method for the cooperative constrained optimization of the power consumption by the network of TCLs controlled by smart power sockets;
    \item A low-cost experimental testbed based on off-the-shelf hardware and purpuse-built software;
    \item An experimental validation of the proposed method.
\end{itemize}

\textbf{Advantages of our multi-agent control architecture:}

\begin{itemize}
\item \textbf{User privacy:} there is no information aggregator or centralized supervisor, information is exchanged only locally by small sets of anonymous users which possibly change over time.

\item \textbf{Plug-and-play architecture:} the proposed architecture is based on smart power sockets which seeminglessly identify and observe the dynamics of the TCL they are connected to.

\item \textbf{Retrofit on existing devices:} smart power sockets are intended to coordinate existing devices, thus greatly reducing the cost of the DSM infrastructure by allowing users to keep their existing domestic appliances.

\item \textbf{Randomized and asynchronous coordination:} The proposed method exploits randomized local updates of the ON/OFF state of SPSs, thus the method can be implemented in large-scale networks without the need for network wide time-synchronization of state updates.

\item \textbf{Robustness to practical implementation issues:} The proposed method exploits a dynamic consensus algorithm for the distributed estimation of the average planned power consumption by the network and randomized local state updates. Therefore, it inherits the robustness properties of dynamic consensus algorithms versus time-varying topologies, communication failures, agent failure, time-delays, changes in the network size. Randomization of state updates prevents the failure of a single agent to disrupt the emergent behavior of the whole network.

\item \textbf{Robustness against Denial of Service (DoS) cyber-attacks:} Since the architecture is based on the anonymous local interactions among devices in an unstructured peer-to-peer network, there is no single computing element that can be attacked to disrupt the network behavior. To deny cooperation among the devices a large-fraction of the large-scale network of devices needs to be attacked, thus increasing significantly the cost and scale of DoS attacks on the network. Furthermore, even disconnecting the network does not prevent the connected components to continue their cooperation to optimize their behavior.
\end{itemize}

\textbf{Structure of the paper:} In Section \ref{sect:modeling} an approximate dynamical model of the generic TCL and SPS device is presented together with the adopted notation. In Section \ref{sect:pr_statement} the  global objective function optimized by the proposed method is presented. In Section \ref{sect:architecture} the Multi-agent Control Architecture is presented, in its subsections the proposed methods for system identification, state estimation, Mixed Logic Dynamical (MLD) modelling of the TCL-plus-SPS system and dynamic consensus algorithms for distribtued estimation are described in detail. In Section \ref{subsect:heuristic} the proposed online distributed optimization method is presented and some of its convergence properties are characterized.
In Section \ref{sect:simulation} a novel low-cost experimental testbed is presented and the proposed online optimization method is validated both numerically and experimentally with real domestic TCLs. Finally, in Section \ref{sec:conclusion} concluding remarks and possible future improvements are discussed.


\section{Modelling of Thermostatically Controlled Loads and Smart Power Sockets}
\label{sect:modeling}

Consider a multi-agent system (MAS) consisting of a population $\mathbfcal{V}=\lbrace 1,\dots,n\rbrace$ of TCLs whose power outlet is plugged into an off-the-shelf SPS adapter. SPSs are provided with a sensor to measure power consumption, processing capability, WiFi communication, and an ON/OFF power switch for either autonomous or remote actuation purposes. Each SPS is connected to a peer-to-peer network over the internet.
Let $\mathbfcal{E}(t_k)\subseteq \left\{\mathbfcal{V}\times\mathbfcal{V}\right\}$ be the set of active communication links among agents, i.e., TCL-plus-SPS systems, at time $t_k$, and let graph $\mathbfcal{G}=(\mathbfcal{V},\mathbfcal{E})$ represent the peer-to-peer network topology. According to \cite{Hao2015189,perfumo2012load}, the dynamics of the $i$-th TCL can be well approximated at discrete time instants $t_1<t_2<\ldots<t_k$, with $k\in \nat$, as follows
\begin{equation}\small
\begin{array}{ccl}
  T^{i}(t_{k+1})&=&T^{i}(t_k)\cdot e^{-\alpha^i \Delta \tau}+\\
&&  \left(1-e^{-\alpha^i \Delta \tau}\right)\left(T_\infty^i+\frac{q^i}{\alpha^i}u^i(t_k)+w(t_k)\right),
\end{array}
  \label{eqn:dyn_TCL}
\end{equation}
where: $\Delta \tau=t_{k+1}-t_k$ is a constant sampling interval; \mbox{$T^{i}(t_k)\in \rea^+$} is the thermostatically controlled  temperature of the TCL; $\alpha^i>0$ denotes the heat exchange coefficient with the operating environment, whose temperature is $T^i_{\infty}$; \mbox{$u^i:\rea^+ \rightarrow\left\{0,1\right\}$} is the $\lbrace \mathrm{OFF},\mathrm{ON}\rbrace$ control signal; $q^i>0$ is the heat generated by the electric heating element. Finally, $w:\rea^+\rightarrow \rea$ models unknown disturbances, e.g., if we consider an electric water heater, it may represent the temperature drop due to the refill process with cold water of the tank after a water drawing event generated by the user. Let $s^i(t_k):\rea^+ \rightarrow \left\{0,1\right\}$ be the binary state associated with the SPS's power switch; since the SPS is connected in series with the TCL's power outlet, follows that $u^i$ can be rewritten as
\begin{equation} \small
u^i(t_k)=s^i(t_k)\cdot h^i(t_{k}),
\label{eqn:ui}
\end{equation}
where $h^i:\rea^+ \rightarrow\left\{0,1\right\}$ is the state of TCL's heater element whose value is updated according to the next thermostatic control rule
\begin{equation} \small h^i(t_{k+1}) : 
\left\{\begin{array}{ccc}
0 & 
\text{if} & T^i(t_k)\geq T_{\mathrm{max}}^i\\
h^i(t_{k}) &
\text{if} & T^i(t_k)\in \left(T_{\mathrm{min}}^i,T_{\mathrm{max}}^i\right)\\
1 & \text{if} & T^i(t_k)\leq T_{\mathrm{min}}^i
\end{array}\right.\label{eqn:u_hyst}
\end{equation}
such that the TCL's temperature remains within the operating range, i.e. $T^i\in\left[T_{\mathrm{min}}^i, T_{\mathrm{max}}^i\right]$. From \eqref{eqn:ui} follows that necessary condition to have $u^i=1$ is that $s^i=1$, but not vice versa.

Let $\mathrm{p}^i>0$ be the nominal amount of power absorbed by the $i$-th TCL in the ON state ($u^i=1$). Since its power consumption is approximately either $0$ or $\mathrm{p}^i$ watts, in the reminder we model the absorbed power profile $p^i(t_k)$ sensed by the SPS, as follows
\begin{equation}\small
p^i(t_k)=  \mathrm{p}^i\cdot u^i(t_k)=\mathrm{p}^i\cdot \left(s^i(t_k)\cdot h^i(t_{k})\right). 
\label{eqn:pi}
\end{equation}

Thus, the  global  instantaneous  power consumption by the network at time $t_k$ is given by
\begin{equation} \small
P(t_k)=\sum_{i\in\mathbfcal{V}} p^i(t_k)=\sum_{i\in\mathbfcal{V}} \mathrm{p}^i\cdot u^i(t_k)=\sum_{i\in\mathbfcal{V}} \mathrm{p}^i\cdot \left(s^i(t_k)\cdot h^i(t_{k})\right)
\label{eqn:Total_Power}
\end{equation}


\section{Coordination Objective for the Multi-Agent System (MAS)}\label{sect:pr_statement}

In this paper we formulate our TCL-plus-SPS control problem described in the introduction as a multi-agent cooperative control problem where the objective is to optimize a globally coupled quadratic cost function over a receding horizon time window which penalizes the peaks of power consumption. Agents are hybrid systems, thus the constraints on their dynamics are local and non-convex. To preserve privacy 
each agent has only access to information on the state of a small set of anonymous neighboring agents, and the timing and order of their state updates is randomized. The P2P network that connects the agents is unstructured and modeled by a connected graph, which is unknown to the agents and possibly time-varying. Cooperation occurs by the distributed estimation of average planned power consumption over the network thorough a dynamic consensus algorithm. The network may be affected by communication failures, changes in topology and size as long as the chosen dynamic consensus algorithm provides bounded tracking error despite these features. 


Let us now formalize the coordination objective of our the problem.
According to \eqref{eqn:dyn_TCL}-\eqref{eqn:pi}, each agent, at time $t_k$, plans its own power consumption within a receding horizon time-window $\mathcal{\tau}(t_k)=\left[t_{k}+ \Delta \tau,t_k+L \Delta \tau\right]$ of length $L \Delta \tau$, which is a schedule vector of $L$ entries, i.e.,
\begin{equation}
\mathbf{s}^i(t_k)=\left[s^i_{1}\cdots s_\ell^i\cdots s_L^i\right]^\intercal =
\left[s^i(t_{k+1})\cdots s^i(t_{k+L})\right]^\intercal
\in\lbrace0,1\rbrace^L  
\label{eqn:s_vector}
\end{equation}
each entry denotes the state that should be actuated by the $i$-th SPS's power switch at the future instants of time $t_{k+\ell}=t_k+\ell \Delta \tau$, $\ell=1,\dots,L$. 
Clearly, \eqref{eqn:s_vector} provides $L$ degrees of freedom to each agent to satisfy its own local constraints on its dynamics while optimizing a global objective function.

In accordance with \eqref{eqn:Total_Power} and the notation in \eqref{eqn:s_vector}, we denote the prediction at time $t_k$ of the planned power consumption by the network at a future time $t_{k+\ell}=t_k+\ell \Delta \tau$ as
\begin{equation}
\small
P_\ell(t_k)=\sum_{i\in\mathbfcal{V}} \mathrm{p}^i \cdot u^i_{\ell}(t_k)=\sum_{i\in\mathbfcal{V}} \mathrm{p}^i\cdot\left( s^i_{\ell}(t_k)h^i_{\ell}(t_k)\right).\label{eqn:Pell} 
\end{equation}


The MAS aims to optimize online through local interactions the next global objective function
\begin{equation}\small \label{eqn:globalObjective}
J(t_k)=\frac{1}{L}\sum_{\ell=1}^{L} \left(P_\ell(t_k)\right)^2=\frac{1}{L}\sum_{\ell=1}^{L} \left(\sum_{i\in \mathbfcal{V}} \mathrm{p}_i \cdot u^i_{\ell}(t_k)\right)^2.
\end{equation}

Notice that, since \eqref{eqn:globalObjective} is a sum over the receding horizon time-window of the squared power consumption expected by the TCLs, it penalizes the peaks power consumption over the window $\tau(t_k)=\left[t_k,t_k+L \Delta \tau\right)$ of length $L \Delta \tau$. Since the total power consumption of the network is squared at each interval, we further point out that $J(t_k)$ is not separable and since the total power consumption of the network is unknown to the agents, the current value of the objective function in \eqref{eqn:globalObjective} is unknown to the agents too.

An intuitive and qualitative way to interpret the objective~\eqref{eqn:globalObjective} is to consider a scenario where an energy provider changes the prices of electricity in an urban area proportionally to the current expected power consumption in a short term horizon in real time, to smooth and shave off peaks in the power demand. If we denote the cost, predicted at time $t_k$, for the electricity at the future time $t_{k+\ell}$ as $c_{\ell}(t_k)=c \cdot P_\ell(t_k)$, $c>0$, then the objective \eqref{eqn:globalObjective} can be reduced 
to 
\begin{equation}\small
J(t_k)=\frac{1}{L} \sum_{\ell=1}^{L} c_{\ell}(t_k) P_\ell(t_k).
\label{eqn:J_tot_example}
\end{equation}
Notice that, since each TCL abides to an embedded constraint to keep its internal temperature at a given threshold, the averaged power demand by the network is not allowed to change significantly. Therefore, by the optimization of function~\eqref{eqn:globalObjective} it is expected, as a byproduct, a reduction in the Peak-to-Average Ratio (PAR) of the power consumption in the network, a metric which is widely accepted as indicating how efficiently the power grid can serve a generic electric load profile \cite{mohsenian2010autonomous,notarnicola2016duality}.

To shape indirectly the electric load profile generated by the MAS to shave the peak load of the corresponding grid network, we introduce "virtual agents" (or virtual loads) which interact with other agents to add a fictitious, possibly large, power consumption to the network at the desired time, forcing the MAS to cooperatively shift the power consumption of each agent away from the considered peak hours. This operation, as opposed to other common approaches to DSM, keeps the agents anonymous, no direct commands are issued by a centralized entity to switch off single systems and the local constraints of operation of each TCL are satisfied, i.e., no TCL is forced to lower its temperature lower than its desired threshold.


Thus said, at each time $t_k$, the generic agent $i$ actuates the plan of its SPS according to the value of $s^i_{1}(t_{k})$ of \eqref{eqn:s_vector}, then it computes a new plan $\mathbf{s}^i$, where from now on we omit the dependence of the planning vector from time $t_k$ for the sake of readability and with no loss of generality. 


In addition, it should be noticed that the objective given in \eqref{eqn:globalObjective} is not a separable function, thus its distributed optimization is more challenging.  Furthermore, it should also be noticed that every agent is subject to a set of local constraints, denoted as $\mathcal{\chi}^i(t_k)$, which are generally non-convex as they involve mixed integer linear constraints (as  it will be discussed in detail in Section \ref{subsect:mld}). This implies that the considered optimization is NP-hard in general. Accordingly, in this paper we propose an heuristic to approximate the optimal solution of the considered problem in a real scenario.

\section{Proposed Multi-agent Control Architecture}\label{sect:architecture}

In this section, we describe the proposed multi-agent control architecture enabled by Smart Power Sockets which add measurement, processing and communication capabilities to traditional domestic TCLs such as electric water heater or electric radiators. In particular, we provide methods for system identification, state observation and Mixed Logic Dynamical (MLD) modeling of the agents representing the combination of the dynamics of a TCL actuated by a SPS.

\subsection{Method for TCL Power Consumption Model Identification\label{subsubsect:Sys_Id}}

A simple method to identify the TCL system parameters would be to measure the internal and external temperature of the device and, by knowing the desired temperature range $[T_{min},T_{max}]$, use any of the the standard system identification methods to approximate the dynamics as a first order system, see \cite{shad2017identification} for the specific case of TCL parameter identification. However, since in our plug-and-play scenario we do not have access to the internal (or external) temperature of the TCL, and the desired temperature range $[T_{min},T_{max}]$ is considered unknown, these methods can not be applied. 


Furthermore, by exploiting only power consumption measurements it is not possible to estimate directly neither the absolute temperature $T^i(t_k)$ controlled by the TCL nor  the maximum and minimum desired temperature $T^i_{\mathrm{min}}$, $T^i_{\mathrm{max}}$ or the ambient temperature $T^i_\infty$. However, by observing the instants of time in which power consumption turns on and off, in accordance with~\eqref{eqn:u_hyst}, the falling (resp. raising) edges on $p^i(t_k)$ in \eqref{eqn:pi} informs us that the temperature $T^i(t_k)$ reached the value of $T^i_{\mathrm{max}}$ (resp. $T^i_{\mathrm{min}}$).
 
Since our objective is to approximate the dynamics of the power consumption of the TCL and not its temperature, by performing a coordinate change in the TCL model we define a so-called ``virtual temperature'', denoted as $y^i(t_k)$,  as follows
\begin{equation}\small
    y^i(t_k)=\beta^i\cdot \left(T^i(t_k)-T^i_\infty\right),
    \label{eqn:yi_change_of_variable}
\end{equation}
where  $\beta^i>0$ is an unknown parameter that  linearly maps $T^i(t_k)\in\left[T_\infty^i,T_{\mathrm{max}}^i\right]$ to $y^i(t_k)\in\left[y_\infty^i=0,y_{\mathrm{max}}^i=1\right]$.

In our system identification procedure we consider the TCL as operating in retention mode, i.e., with no significant temperature drops due to exogenous disturbances, i.e. $w^i\equiv 0$ in \eqref{eqn:dyn_TCL}, and $s^i(t_k)=1$ $\forall$ $t_k\geq0$. By substituting \eqref{eqn:yi_change_of_variable} into \eqref{eqn:dyn_TCL}, with $w^i\equiv0$, it yields our power consumption model
\begin{equation} \small
y^i(t_{k+1})=y^i(t_k) \cdot e^{-\alpha^i \Delta \tau}+\frac{\beta^i q^i}{\alpha^i}\left(1-e^{-\alpha^i \Delta \tau}\right) u^i(t_k).
\label{eqn:dyn_y}
\end{equation}

Let us now define ${\Delta T}^i_{\mathrm{off}}$ (resp. ${\Delta T}^i_{\mathrm{on}}$) as the discharge (resp. charge) interval of time to be waited in order to bring $T^i$ from $T^i_{\mathrm{max}}$ (resp. $T^i_{\mathrm{min}}$) to $T^i_{\mathrm{min}}$ (resp. $T^i_{\mathrm{max}}$) when $w^i(t_k)=0$, $\forall$ $t_k\geq0$.

If $T^i(t_k)=T^i_{\mathrm{max}}$ (i.e. \mbox{$y^i(t_k)=1$}) at $t_k=0$, from \eqref{eqn:dyn_y} it results
\begin{equation}\small
\small y^i(\Delta T^i_{\mathrm{off}})=y^i_{\mathrm{max}}\cdot e^{-\alpha^i \Delta T^i_{\mathrm{off}}}=y^i_{\mathrm{min}}\rightarrow \alpha^i=\frac{-ln\left(y^i_{\mathrm{min}}\right)}{\Delta T^i_{\mathrm{off}}}.
\label{eqn:alpha_i}
\end{equation}
On the contrary, if $T^i(t_k)=T^i_{\mathrm{min}}$ at $t_k=0$, from \eqref{eqn:u_hyst} it yields
\begin{equation}
\small y^i(\Delta T^i_{\mathrm{on}})=y^i_{\mathrm{min}}\cdot e^{-\alpha^i  \Delta T^i_{\mathrm{on}}}+\frac{\beta^i q^i}{\alpha^i}\left(1-e^{-\alpha^i \Delta T^i_{\mathrm{on}}}\right)=y^i_{\mathrm{max}}.
\label{eqn:y_charge}
\end{equation}
Then, by substituting \eqref{eqn:alpha_i} into \eqref{eqn:y_charge}, with $y^i_{\mathrm{max}}=1$, we have
\begin{equation}\small
    \frac{\beta^i q^i}{\alpha^i}=\frac{1 - y^i_{\mathrm{min}}\cdot e^{-\alpha^i\Delta T^i_{\mathrm{on}}}}{1 - e^{-\alpha^i\Delta T^i_{\mathrm{on}}}}=\frac{1 - y^i_{\mathrm{min}}\cdot e^{ln(y^i_{\mathrm{min}})\frac{\Delta T^i_{\mathrm{on}}}{\Delta T^i_{\mathrm{off}}}}}{1 - e^{ln(y^i_{\mathrm{min}})\frac{\Delta T^i_{\mathrm{on}}}{\Delta T^i_{\mathrm{off}}}}}
    \label{eqn:K_i}
\end{equation}

The intervals of time  $\Delta T^i_{\mathrm{off}}$ and $\Delta T^i_{\mathrm{on}}$ can be easily measured by the SPS power measurements.  On the other hand, parameter $y^i_{\mathrm{min}}$ is neither known nor available from measurements, and the solution of \eqref{eqn:alpha_i}, \eqref{eqn:K_i} is not unique. However, $y_{\mathrm{min}}$ can be tuned experimentally within the open domain $(y^i_{\infty},y^i_{\mathrm{max}})\equiv(0,1)$ {should we reverse $y^i_{\infty},y^i_{\mathrm{max}}$ here?} to get the best approximation of the TCL power consumption dynamics. In our model $y_{\mathrm{min}}$ is a free choice parameter. By its choice the other parameters of our model are computed. In our numerical and experimental validation we choose $y_{\mathrm{min}}=0.5$ and identify the other parameters in the model accordingly. 



To properly evaluate $\Delta T^i_{\mathrm{off}}$ and $\Delta T^i_{\mathrm{on}}$ from the SPS's power measurements, due to the time-varying behaviour of the ambient temperature $T^i_{\infty}$ during the day, as well due to unmodeled dynamics or disturbances, a single charge-discharge cycle is not sufficient to get a statistically robust information on these quantities. Experimental tests indicate that a period of time including four or five charge-discharge cycles is enough to obtain a good statistical information. In this regard, on the top of Figure~\ref{fig:identification} a ten hours power consumption profile from a real electric water heater acquired during the night hours, where no water drawing events occur, is shown. In particular, let $\Delta T^i_{\mathrm{on},k}$ and $\Delta T^i_{\mathrm{off},k}$ be the $k$-th charge and discharge time intervals of the identification data-set of power measurement, we considered the median value of the dataset
\begin{equation}\small
\begin{bmatrix}
{\Delta T}^i_{\mathrm{on}}\\
{\Delta T}^i_{\mathrm{off}}
\end{bmatrix}=\mathrm{Median}
\begin{pmatrix}
\begin{bmatrix}\Delta T^i_{\mathrm{on},1}&\cdots& \Delta T^i_{\mathrm{on},k}& \cdots&\\
\Delta T^i_{\mathrm{on},1}&\cdots& \Delta T^i_{\mathrm{on},k}& \cdots&
\end{bmatrix}
\end{pmatrix}
    \label{eqn:charge_discharge_intervals}
\end{equation}

\begin{figure}[!t]
    \centering
    \includegraphics[width=21pc]{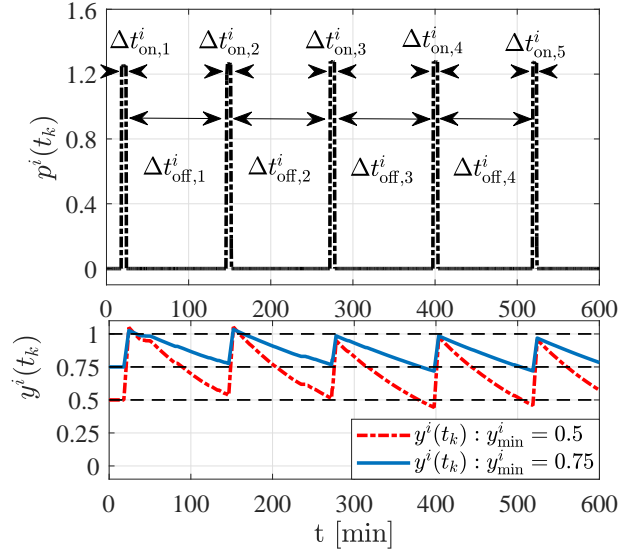}
    \caption{ Top: Power consumption of a domestic electric water heater used for experimental tests during ten hours of operation, at night. Bottom: Output of model \eqref{eqn:dyn_y} corresponding to the power consumption shown on top, with the parameters estimated with two different choices of $y^i_{\mathrm{min}}$.}
    \label{fig:identification}
\end{figure}

\begin{table*}[!tbp]
    \caption{Description of discrete events of the hybrid virtual temperature observer depicted in Figure \ref{fig:hybrid_obs}}
    \label{tab:observerEvents}
    \centering
    \begin{tabular}{|c|c|c|}
    \toprule
      \textbf{Events}   &  \textbf{Triggering conditions} & \textbf{Effects on continuous states}\\
      \midrule
       $\mathrm{E_{off}}$      &   $\left( s^i(t_{k-2})=1 \right)\ \band \ \left( s^i(t_{k-1})=1\right) \band \left( s^i(t_{k})=1\right) \band \ \left( u^i(t_{k-1})=1\right) \band \left( u^i(t_{k})=0\right)$ & $\left(y^i(t_k):=y^i_{\mathrm{max}}~,~T_r^i:=0\right)$    \\
       \midrule
       $\mathrm{E_{on}}$       &  $\left( s^i(t_{k-2})=1\right) \band \left( s^i(t_{k-1})=1 \right)\band \left( s^i(t_{k})=1 \right)\band \left( u^i(t_{k-1})=0\right) \band \left( u^i(t_{k})=1\right)$ & $\left(y^i(t_k):=y^i_{\mathrm{min}}~,~T_r^i:=0\right)$  \\
       \midrule
       Timeout & $T_r^i\geq \bar{T}^i_r$     &  \text{none}   \\
       \midrule
        Out of Bound 1 &  $y^i(t_k)< \beta_1\cdot y^i_{\mathrm{min}}$ &  $y^i(t_k):=\beta_1\cdot y^i_{\mathrm{min}}$  \\
       \midrule
        Out of Bound 2 &  $y^i(t_k)> \beta_2\cdot y^i_{\mathrm{max}}$ &  $y^i(t_k):=\beta_2\cdot y^i_{\mathrm{max}}$  \\
       \bottomrule
    \end{tabular}
\end{table*}

\begin{figure}[!tbp]
\centering
\includegraphics[width=20pc]{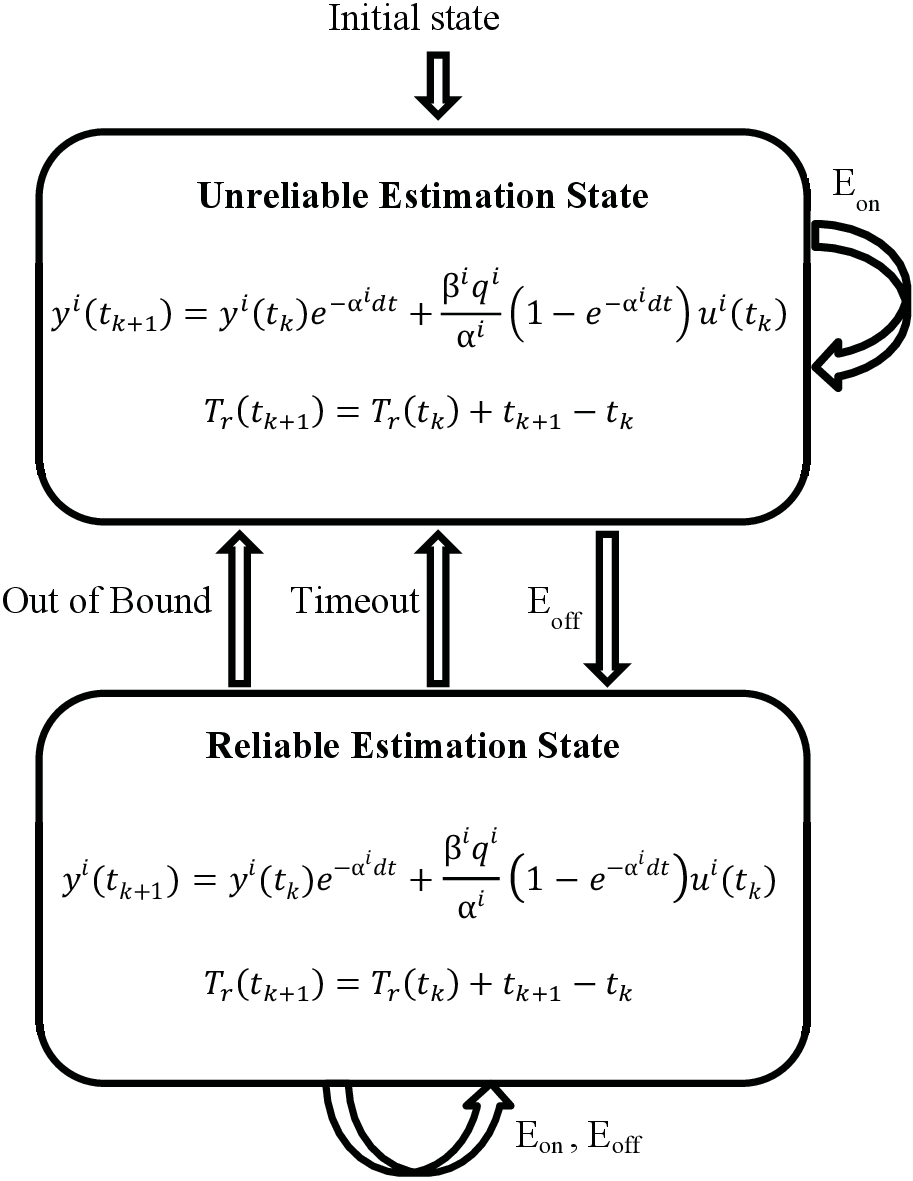}
\caption{Local hybrid virtual temperature observer.}
\label{fig:hybrid_obs}
\end{figure}

\textbf{Experimental Validation: }
On the bottom of Figure~\ref{fig:identification} it is shown the force response of model \eqref{eqn:dyn_y} with input corresponding to the measured power consumption shown on top of Figure~\ref{fig:identification}, with the parameters estimated with \eqref{eqn:charge_discharge_intervals} by considering $5$ charge/discharge cycles,\eqref{eqn:alpha_i}, \eqref{eqn:K_i}, and for two different values of $y^i_{\mathrm{min}}$. It can be seen that the method is sufficiently accurate for our purposes. It can be further noted that the choice of $y^i_{\mathrm{min}}$ does not affect the duration of the charge-discharge cycle. It follows that the estimation of the virtual temperature may be used to predict the future power consumption of the device. {\hfill $\blacksquare$}

\vspace{0.1cm}

\subsection{Hybrid Virtual Temperature Observer\label{subsubsect:Observer}}
We now consider the problem of estimating the virtual temperature $y^i(t_k)$ in \eqref{eqn:yi_change_of_variable} by exploiting only the SPS measurements. The discrete state of the power switch of the SPS is known and is denoted by $s^i(t_k)$, the state of the heater element in the TCL, $u^i(t_k)=s^i(t_k)h^i(t_{k})$ , can be evaluated by applying a threshold to the measurement on the current power consumption as follows
\begin{equation}\small
u^i(t_k) : \left\lbrace
\begin{array}{cl}
     1 & \mathrm{if}~ p^i(t_k)>\epsilon^i_{\mathrm{n}} \\
     0 & \mathrm{otherwise}
\end{array}
\right.
    \label{eqn:ui_threashold}
\end{equation}
where $\epsilon_{\mathrm{n}}^i>0$ denotes the noise measurement level on $p^i(t_k)$. 
         
The block diagram of the proposed observer is depicted in Figure~\ref{fig:hybrid_obs}. It is a hybrid system with a continuous state $y^i(t_k)$ representing the estimated virtual temperature, a timer $T_r^i(t_k)$ and two discrete states, referred as the ``Unreliable Estimation state'' (UE) and ``Reliable Estimation state'' (RE). Table~\ref{tab:observerEvents} lists the events that trigger the transition between the discrete states and their set and reset effects on $y^i(t_k)$ and $T_r^i$. The symbol $\band$ denotes the ``and'' logical connective operator.

The UE states denote the case in which no reliable prediction on $y^i(t_k)$ can be done due to the uncertainties on model~\eqref{eqn:dyn_TCL}, and the presence of exogenous unpredictable disturbances $w^i(t_k)$.The RE state denotes the case where the estimation on $y^i(t_k)$ is sufficiently accurate to be used to predict the behavior of the TCL with a small error.  At inizialization, the observer is set to the UE state.

By \eqref{eqn:ui} and \eqref{eqn:u_hyst}, whenever the TCL control logic autonomously changes state from $u^i(t_{k-1})=1$ to $u^i(t_{k})=0$, it can be inferred that its controlled temperature $T^i(t_k)$ has reached its upper limit $T^i_{\mathrm{max}}$ which corresponds to a virtual temperature of $y^i(t_k)=1$. Similarly, we can gather useful information when the control input changes from $u^i(t_{k-1})=0$ to $u^i(t_{k})=1$, i.e., $y^i(t_k)=y^i_{\mathrm{min}}$.

In accordance with Table~\ref{tab:observerEvents} we refer to these events, resp., as ``$\mathrm{E_{off}}$'' and ``$\mathrm{E_{on}}$''. It is worth to point out that, whenever these events are triggered the estimated virtual temperature is set to the corresponding correct value, thus resetting any error due to uncertainties or disturbances. We thus refer to these events also as ``synchronization events''. 


{From the UE state, whenever ``$\mathrm{E_{off}}$'' is triggered, i.e., the SPS is ON and the device autonomously switches from the ON state with $u^i(t_k)=1$ to OFF with $u^i(t_k)=0$, we are able to guarantee that the current estimation is reliable and will remain so until a ``Timeout'' or an ``Out of Bound'' event occurs.

The ``Timeout'' event is triggered whenever the clock variable $T_r^i$ measures that an interval of time greater than $\bar{T}_r^i$ has passed since the last synchronization event. The ``Timeout'' event is needed because in between synchronization events the estimation of the virtual temperature is open loop, thus we set the observer to the UE state if no synchronization events occur within a maximum time window due to an expected large estimation error. In accordance with the notation in Figure~\ref{fig:identification}, $\bar{T}_r^i$ is experimentally set equal to three times the maximum $\Delta T^i_{\mathrm{off},k}$ recorded during the TCL parameter identification procedure.

The ``Out of Bound'' event is triggered whenever the estimated virtual temperature exceeds its expected bounds by certain percentage. The occurrence of this event is an indication that either a significant temperature drop is occurred (e.g. due to the refill process with cold water of a boiler after an water drain event) or that model uncertainties were become significant that the current estimation is no more reliable.

Finally, we notice that when the event ``$\mathrm{E_{on}}$'' is triggered from the UE state, we do not consider the new value of the estimated virtual temperature to be reliable.
This is because strong exogenous disturbances, such as drawing hot water from the TCL, might have brought its controlled temperature much lower than $T^i_{min}$ and therefore the autonomous switching to the ON state does guarantee by itself that the the TCL temperature corresponds to $T^i_{min}$ in that particular instant of time.  

{\textbf{Experimental Validation:}} 
An experimental validation of the proposed hybrid observer is shown in Figure~\ref{fig:real_TCL}, the test was performed with a real domestic water heater with parameters identified in accordance with Subsection~\ref{subsubsect:Sys_Id} with $\epsilon^i_{\mathrm{n}}=5$W, $\beta_1=0.8$ and $\beta_2=1.2$ and affected by disturbances in the form of water drawing events and model uncertainties. In the test the SPS was always in the ON state. We can observe that the virtual temperature estimation is accurate, and thus reliable, from the first falling edge on $u^i(t_k)$ up to $\approx210$min, and from $350$min to the end of the test, since a sufficient number of synchronization events ``$\mathrm{E_{on}}$'' and ``$\mathrm{E_{on}}$'' occurs and we do not see large errors in the estimated temperature profile. On the contrary, at around $210$min the observer detects an ``Out of Bound 2'' event, due to the drawing of hot water from the water heater. Because of that, the observer sets itself to the UE state while saturating $y^i(t_k)$ at its maximum admissible value. Then, at $250$min, a ``$\mathrm{E_{off}}$'' event is triggered and the observer sets itself  back to the RE state. From $\approx250$ min to $\approx350$ min,  the observer is on the RE state, in this case estimation errors seems higher since $u^i(t_k)$ switches ON earlier than the expected time ($\approx\Delta T^i_{\mathrm{off}}$). This is due to cold water refilling the tank and heat exchange in fluids occurring by convection currents \cite{kakacc1991boilers}, not modeled in neither \eqref{eqn:dyn_TCL} nor in \eqref{eqn:dyn_y}. However, despite model uncertainties, since the synchronization events reset $y^i(t_k)$ to the correct value, the estimation error is kept small and bounded in an actual real scenario. 

Finally, in Figure \ref{fig:real_TCL2} it is shown the evolution of the observer state during the execution of the TCL cooperation protocol for a real TCL during the experimental test described in Section \ref{sect:simulation}. It can be seen that, despite model uncertainties and disturbances, the occurrence of synchronization events keeps estimation errors bounded.  
{\hfill $\blacksquare$}}

\begin{figure}[!t]
\centering
\includegraphics[width=20pc]{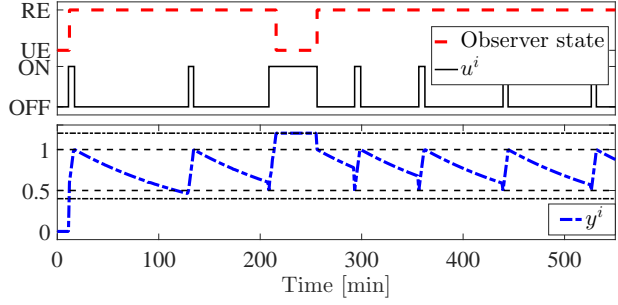}
\caption{Experimental validation of the hybrid observer with a real water heater. Top: evolution of the discrete state of the observer and the TCL heating element. Bottom: evolution of estimated virtual temperature.}
\label{fig:real_TCL}
\end{figure}
\begin{figure*}[!t]
\centering
\includegraphics[width=42pc]{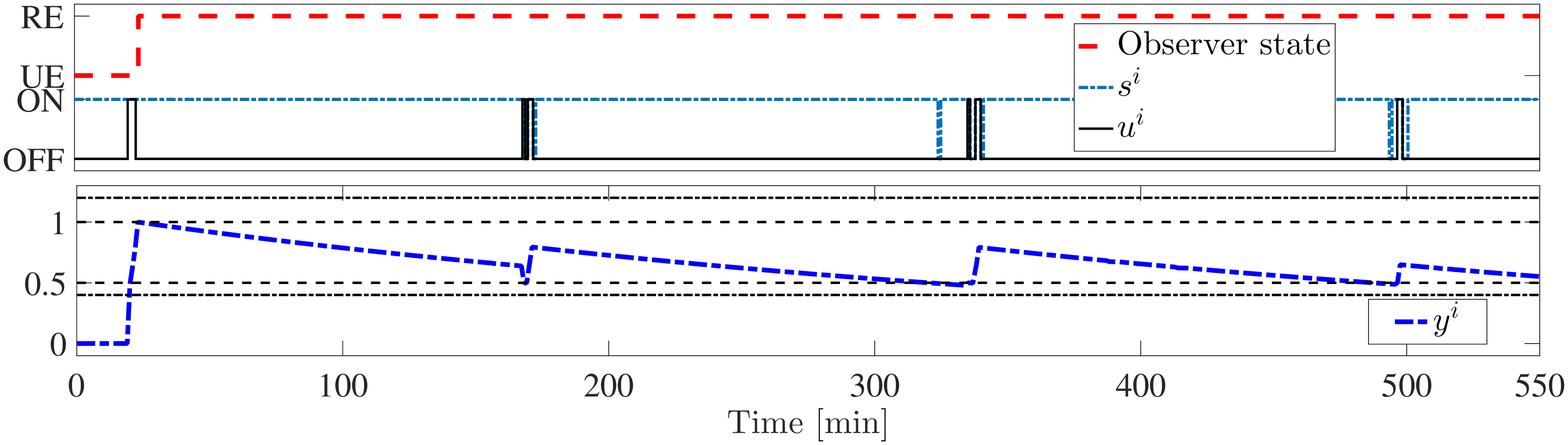}
\caption{Experimental validation of the hybrid observer with a real water heater. Top: evolution of the discrete state of the observer and the TCL heating element and SPS state during the execution of the TCL cooperation protocol. Bottom: evolution of estimated virtual temperature.}
\label{fig:real_TCL2}
\end{figure*}


\subsection{Mixed Logic Dynamical modeling of TCL-plus-SPS agents}\label{subsect:mld}

{To enable the cooperative optimization among networked TCL-plus-SPS agents,
we model the local hybrid dynamic \eqref{eqn:dyn_TCL}-\eqref{eqn:u_hyst} as a Mixed Logic Dynamical (MLD) system, see \cite{bemporad1999control} for an overview on the modelling methodology. 
Briefly, MLD exploits propositional calculus to identify a set of linear integer constraints that characterize the dynamical behavior of a hybrid system and enables the implementation of model predictive control by linear programming. 

We now translate the operative constraints \eqref{eqn:dyn_TCL}-\eqref{eqn:u_hyst} of the generic agent $i$ at time $t_k$, into a set of mixed integer linear constraints denoted as $\mathcal{\chi}^i(t_k)$.

A model of the normalized virtual temperature ${y^i(t_k)\in[0,1]}$, as defined in Section~\ref{subsubsect:Sys_Id}, is exploited instead of the physical temperature $T^i(t_k)\in[T_{\infty}^i, T_{\mathrm{max}}]$ because the considered SPS can measure only the electric power absorbed by the TCL while its internal temperature can not be measured. 

However, it is worth to remark that if $T^i(t_k)$ were available, then the following MLD characterization of the TCL-plus-SPS system would still hold without changes, by simply considering $T^i(t_k)$ in place of $y^i(t_k)$.

Thus, let us consider the dynamics of the virtual temperature \eqref{eqn:dyn_y} of the $i$-th TCL, and let \mbox{$A_i= e^{-\alpha^i\Delta \tau}$} and \mbox{$B_i=\frac{\beta^i q^i}{\alpha^i}(1-e^{-\alpha^i\Delta \tau})$}.} Similarly, following the notation in \eqref{eqn:s_vector}, { let us further define}
\begin{equation}\label{eqn:u_vector}
{\small
\begin{array}{ll}
\mathbf{u}^i(t_k)&=\left[u^i_{1}\cdots u_\ell^i\cdots u_L^i\right]^\intercal  \\
&=\left[u^i(t_{k+1}) \cdots u^i(t_{k+\ell}) \cdots u^i(t_{k+L})\right]^\intercal,\\
\end{array}
}
\end{equation}
\begin{equation}\label{eqn:s_vector}
{\small
\begin{array}{ll}
\mathbf{y}^i(t_k)&=\left[y^i_{1}\cdots y_\ell^i\cdots y_L^i\right]^\intercal  \\
&=\left[y^i(t_{k+1})\cdots y^i(t_{k+\ell}) \cdots y^i(t_{k+L})\right]^\intercal, 
\end{array}
}
\end{equation}
whose entries denote the control input and the virtual temperature over the prediction horizon from time $t_{k+1}$ to $t_{k+L}$.

From \eqref{eqn:dyn_y} it yields $y^i_1={A}_i y^i(t_k)+{B}_i u^i(t_k)$, thus the predicted profile of the virtual temperature from time $t_{k+2}$ to time $t_{k+L}$, evaluated at time $t_k$ as function of $\mathbf{u}^i(t_k)$, is
\begin{equation}\label{eqn:y_update_big}
\renewcommand*{\arraycolsep}{1pt}
\small
  {
  \begin{bmatrix} \small
    y_{2}^i \\
    y_{3}^i\\
    \vdots \\
    y_{L}^i
  \end{bmatrix}
  }
  =
  \underbrace{
  \begin{bmatrix} \small
    {B}_i &  0 & \cdots & 0 \\
    {A}_i{B}_i & {B}_i & \cdots & 0 \\
    \vdots & & \ddots\\
    {A}^{L-{ 2}}_i{B}_i  & {A}^{L-{3}}_i{B}_i  & \cdots & {B}_i
  \end{bmatrix}
  }_{    \mathbf{F}^i}
  \underbrace{
  \begin{bmatrix} \small
    u_{1}^i \\
    u_{2}^i\\
    \vdots \\
    u_{L}^i
  \end{bmatrix}
  }_{\mathbf{u}^i(t_k)}
  +\underbrace{
  \begin{bmatrix} \small
    {A}_i \\
    {A}^{2}_i\\
    \vdots \\
    {A}^{L{ -1}}_i
  \end{bmatrix}
  }_{\mathbf{G}^i} {\small y^{i}_1}.
\end{equation}


For completeness sake, by noting that $y^i_{1}$ is fully determined by $u^i(t_k)$ and $y^i(t_k)$ which are not decision variables in the interval $\left[t_k,t_k+\Delta \tau\right)$, according to \eqref{eqn:y_update_big}, it holds
{
\begin{equation}\small
    \mathbf{y}^i(t_k)=
    \begin{bmatrix}
    0 & \bm{0}\\
    \bm{0} & \mathbf{F}^i
    \end{bmatrix}\begin{bmatrix}
     0\\
     \mathbf{u}^i(t_k)
    \end{bmatrix}+
    \begin{bmatrix}
     1\\
    \mathbf{G}^i
    \end{bmatrix}({B}_i u^i(t_k)+ {A}_i y^i(t_k)).
\label{eq:y_vector_set}
\end{equation}}
Let us now discuss how to model the hybrid behaviour of the TCL-plus-SPS system \eqref{eqn:ui}-\eqref{eqn:u_hyst} into a set of linear integer constraints. First, we define
\begin{equation}\label{eqn:f_g}\small
\begin{aligned}
g_{\ell}^i &\triangleq y_{\ell}^i - y_{\min}^i,\\
f_{\ell}^i &\triangleq y_{\ell}^i - y_{\max}^i.\\
    \end{aligned}
\end{equation}
Clearly, $g_{\ell}^i\geq 0$ implies $y_{\ell}^i\geq y_{\min}^i$ and $f_{\ell}^i\geq 0$ implies $y_{\ell}^i\geq y_{\max}^i$. Let us define the logical \emph{connectives} ``$\band$'' (and), ``$\bor$'' (or), ``$\bnot$'' (not), $\bimplies$ (implies) and $\biff$ (if and only if). 

We associate two dummy binary variables $\delta_{1,\ell}^i$, $\delta_{2,\ell}^i \in\left\{0,1\right\}$ for each $\ell$, to the next inequalities
\begin{gather}
\small
    [\delta^i_{1,\ell}=1]  \biff [g^i_{\ell}\leq 0],   \label{eqn:delta1}\\
     \small[\delta^i_{2,\ell}=1] \biff  [f_{\ell}^i \leq 0].
     \label{eqn:delta2}
\end{gather}
If, for instance, $y_{\ell}^i\in \left(y^i_{min},y^i_{max}\right]$ then $\delta^i_{1,\ell}=0$ and $\delta_{2,\ell}^i=1$. 
\begin{figure}
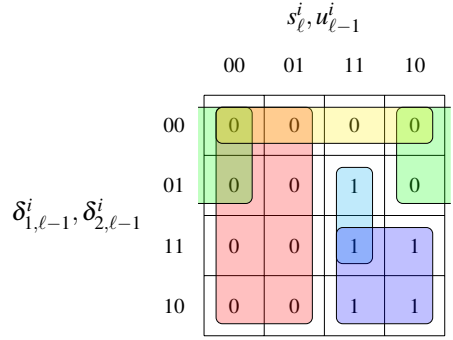

\centering
\scalebox{.8}{
    \begin{karnaugh-map}[4][4][1][\large $s_{\ell}^i,u_{\ell-1}^i$][\large $\delta^i_{1,\ell-1},\delta^i_{2,\ell-1}$]
        \manualterms{0,0,0,0,0,0,0,1,0,0,1,1,0,0,1,1}
        \implicant{0}{9}
        \implicantedge{0}{4}{2}{6}
        \implicant{0}{2}
        \implicant{7}{15}
        \implicant{15}{10}
    \end{karnaugh-map}}
    \caption{Karnaugh map of $u_{\ell}^i$ as function of $u_{\ell-1}^i$, $s_{\ell}^i$, $\delta^i_{1,\ell-1}$ and $\delta^i_{2,\ell-1}$. }\label{fig:map_table}
\end{figure}
    
In Figure~\ref{fig:map_table} we compute the Karnaugh map of the boolean variable $u^i_{\ell}$ based on the control logic \eqref{eqn:ui}-\eqref{eqn:u_hyst}, and the values of $u_{\ell-1}^i$ , $s^i_{\ell}$ and $\delta^i_{1,\ell-1}$, $\delta^i_{2,\ell-1}$. Notice that the knowledge of $y^i_{\ell}$, in accordance with \eqref{eqn:delta1}-\eqref{eqn:delta2}, provides an unique correspondence to $\delta^i_{1,\ell}$ and $\delta^i_{2,\ell}$. For the sake of clarity, we write the equivalent logical statements associated with Figure~\ref{fig:map_table}, as    
\begin{eqnarray} 
{\small [s^i_{\ell}=0]   \bimplies  [u^i_{\ell}=0] , \label{eqn:constraint_1}}\\
{\small   [\delta^i_{1,\ell-1}=0] \band  [u^i_{\ell-1}=0] \bimplies  [u^i_{\ell}=0],\label{eqn:constraint_2}}\\
{\small [\delta^i_{1,\ell-1}=0] \band [\delta^i_{2,\ell-1}=0]  \bimplies  [u^i_{\ell}=0],\label{eqn:constraint_3}}\\
{\small [s^i_{\ell}=1]  \band [u^i_{\ell-1}=1] \band [\delta^i_{2,\ell-1}=1] \bimplies [u^i_{\ell}=1],
\label{eqn:constraint_4}}\\
{\small     [s^i_{\ell}=1]  \band [\delta^i_{1,\ell-1}=1]  \bimplies  [u^i_{\ell}=1].\label{eqn:constraint_5}}
\end{eqnarray}

In particular, \eqref{eqn:constraint_1} is derived from \eqref{eqn:ui}, and it states that $s^i_{\ell}=0$ implies $u^i_{\ell}=0$, but not viceversa. This implication is shown in Figure~\ref{fig:map_table},  associated to the $\mathrm{4\times 2}$ implicant highlighted in red. Statement \eqref{eqn:constraint_2} is derived from \eqref{eqn:u_hyst}, it states that $u^i_{\ell}=0$ if $u^i_{\ell-1}=0$, and $y_{\ell-1}^i>y^i_{\mathrm{min}}$ (thus $\delta^i_{1,\ell-1}=0$);
in Figure~\ref{fig:map_table}, \eqref{eqn:constraint_2} is associated to the quadratic green implicant. Similarly, \eqref{eqn:constraint_3} states that $u^i_{\ell}=0$ if $y^i_{\ell-1}>y^i_{\mathrm{max}}$ (thus $\delta^i_{1,\ell-1}=\delta_{2,\ell-1}^i=0$); in Figure~\ref{fig:map_table}, \eqref{eqn:constraint_3} is associated to the $\mathrm{1\times 4}$ implicant highlighted in yellow. Finally the statements \eqref{eqn:constraint_4} and \eqref{eqn:constraint_5} are associated to, resp., the $\mathrm{2\times 1}$ cyan and $\mathrm{2\times 2}$ violet implications. In particular, they state that $u^i_{\ell}=1$ if $s^i_{\ell}=1$, $u^i_{\ell-1}=1$ (thus $h_{\ell-1}^i(t_k)=1$, i.e., the heating is ON)), $y^i_{\ell-1}\leq y^i_{\mathrm{max}}$ (thus $\delta_{2,\ell}^i(t_k)=1$), or resp., $s^i_{\ell}=1$ and $y^i_{\ell-1}<y^i_{\mathrm{min}}$ (thus $\delta_{1,\ell-1}^i=1$)). 

Although \eqref{eqn:constraint_4} and \eqref{eqn:constraint_5} require the TCL to be ON, $u^i_{\ell}=1$, whenever $y^i_{\ell-1}<y^i_{\mathrm{min}}$ (thus $\delta_{1,\ell-1}^i=1$), we have to include a further constraint to force the SPS to be ON, $s^i_{\ell}=1$, and thus allow the TCL to switch ON whenever necessary, as follows
\begin{align}\small
        [\delta^i_{1,\ell-1}=1] \bimplies  [s^i_{\ell}=1].
    \label{eqn:constraint_6}
\end{align}

Let us now translate \eqref{eqn:delta1}-\eqref{eqn:constraint_6} to a set of mixed-integer linear inequalities, according to the MLD paradigm, which can be solved by Mixed-Integer linear Programming. Firstly, note that the virtual temperature vector $\mathbf{y}^i(t_k)$ is bounded element-wise, thus \eqref{eqn:f_g} admits minimum and maximum points defined as follows
\begin{equation}\label{eqn:max_min_f_g}\small
\begin{array}{c}
m_{f}^i  = \min\limits_{\ell=1,\ldots,L}\;f_{\ell}^i, \quad\quad m_{g}^i  = \min\limits_{\ell=1,\ldots,L}\;g_{\ell}^i,\\
M_{f}^i  = \max\limits_{\ell=1,\ldots,L}\;f_{\ell}^i, \quad\quad M_{g}^i  = \max\limits_{\ell=1,\ldots,L}\;g_{\ell}^i.
\end{array}
\end{equation}

Then, by invoking \cite[\mbox{Properties (2d) and (4e)}]{bemporad1999control}, the logical implications in \eqref{eqn:delta1}-\eqref{eqn:constraint_6} can be formulated as the next set of linear mixed inequalities for $\ell=1,\ldots,L$
\begin{align}\label{eqn:X_relay:eq1_X1}
\eqref{eqn:delta1}~\equiv~&{\small \left\lbrace\begin{array}{rcl}
     g^i_{\ell}&\leq& M_g^i\cdot(1-\delta_{1,\ell}^i)  \\
     g^i_{\ell}&\geq& m_g^i\cdot \delta_{1,\ell}^i
\end{array}\right.}\\
\label{eqn:X_relay:eq3_X1}
    \eqref{eqn:delta2}~\equiv~&{\small\left\lbrace\begin{array}{rcl}
    f^i_{\ell}&\leq& M_f^i\cdot(1-\delta_{2,\ell})\\
    f^i_{\ell}&\geq& m_f^i\cdot \delta_{2,\ell}\\
\end{array}\right.}\\
\label{eq19}
\eqref{eqn:constraint_1}~\equiv~&{\small s^i_{\ell}} \geq  u^i_{\ell}\\
\label{eq21}
\eqref{eqn:constraint_2}~\equiv~&{\small u^i_{\ell}\geq \delta_{1,\ell-1}+u^i_{\ell-1}}\\
\label{eq20}
\eqref{eqn:constraint_3}~\equiv~&{\small u^i_{\ell}\geq  \delta_{1,\ell-1}+\delta_{2,\ell-1}}\\
\label{eq17}
\eqref{eqn:constraint_4}~\equiv~&{\small u^i_{\ell}\geq s^i_{\ell}+u^i_{\ell-1}+\delta_{2,\ell-1}-2}\\
\label{eq18}
\eqref{eqn:constraint_5}~\equiv~&{\small u^i_{\ell}\geq s^i_{\ell}+\delta_{1,\ell-1}-1}\\
\label{eq22}
\eqref{eqn:constraint_6}~\equiv~&{\small s^i_{\ell}\geq \delta_{1,\ell-1}}.
\end{align}

Now, experimental tests have shown that the performance of the proposed hybrid observer is greater if the SPS is ON most of the time. Thus, we introduce a constraint which forces the SPS to be ON for at least a minimum fraction $S^i_{\mathrm{\%,on}}$, of the receding horizon time-window $\tau(t_k)$. In particular, let $S^i_{\mathrm{\%,on}}=\frac{r}{L}$ where $r\in \left[1,L\right]$ is an integer, then 
\begin{equation}
\small
    \sum_{\ell=1}^L s^i_{\ell}\geq L \cdot {S}^i_{\mathrm{\%,on}}
    \label{eqn:s_limit_off}
\end{equation}

Thanks to \eqref{eqn:s_limit_off}, the number of occurrences of synchronization events $\mathrm{E_{on}}$ and $\mathrm{E_{off}}$ in the hybrid observer during the experiments has been greater and thus its corresponding estimation is reliable more often, improving the performance of the overall control architecture.


Finally, by combining \eqref{eqn:y_update_big}, and \eqref{eqn:X_relay:eq1_X1}-\eqref{eqn:s_limit_off} with $\ell=1,\ldots,L$, we get the set of local MLD constraints $\mathcal{\chi}^i(t_k)$  associated to each agent $i$ in the DSM optimization problem of the functional \eqref{eqn:globalObjective}. 

\begin{definition}[Local constraint set $\mathcal{\chi}^i(t_k)$] \label{constraintDefinition}
The SPS control action by each agent $i$ is constrained by a set of linear mixed inequalities denoted as $\mathcal{\chi}^i(t_k)$, which includes inequalities from \eqref{eqn:X_relay:eq1_X1} to \eqref{eqn:s_limit_off} and the equality \eqref{eqn:y_update_big} 
with the addition of $\mathbf{s}^{i}_{\ell}, \mathbf{u}^{i}_{\ell}, \mathbf{\delta}_{1,\ell}^{i}, \mathbf{\delta}_{2,\ell}^{i} \in \left\{0,1\right\}$ and $y^i_{\ell}\in \rea^+$ for $\ell=1,\ldots,L$. \hfill $\blacksquare$
\end{definition}

\medskip

Note that set $\mathcal{\chi}^i(t_k)$ consists of $11\cdot L+1$ linear inequalities and $5L$ variables but
only the $L$ elements of vector $\mathbf{s}^{i}(t_k)$  correspond to actual decision variables. Indeed, if the SPS control actions $\mathbf{s}^{i}(t_k)$ as defined in \eqref{eqn:s_vector} are given, all other variables are uniquely defined by~$\mathcal{\chi}^i(t_k)$.

\subsection{Protocols for Dynamic Average Consensus}\label{subsect:dyn_consensus}

One of the key ideas of the proposed architecture is to exploit local interaction protocols for dynamic average consensus to enable each agent to estimate online the profile of the time-varying future planned global average power consumption $P_\ell$ in \eqref{eqn:Pell} of the TCLs. The estimation computed at time $t_k$ of the average power consumption at time $t_{k+\ell}=t_{k}+\ell \Delta \tau$, is
\begin{equation}\small
\bar{P}_{\ell}(t_k)=\frac{1}{n}\cdot P_\ell(t_k)
=\frac{1}{n}\sum_{i\in\mathbfcal{V}} p^i_\ell(t_k)=\frac{1}{n}\sum_{i\in\mathbfcal{V}}\mathrm{p}^i\cdot u^i_{\ell},~\ell=1,\dots,n,
    \label{eqn:Pell_bar}
\end{equation}
where $p^i_\ell(t_k)=\mathrm{p}^i\cdot{u}_{\ell}^i(t_k)$ is the  power consumption plan of the $i$-th agent, available at time $t_k$. 
In the reminder ${p}^i_\ell(t_k)$ is considered as the local time-varying reference input of each agent that executes the dynamic consensus algorithm. 

In particular, let $\bar{P}^i_\ell(t_k)$ be the estimation computed at time $t_k$ of the average power consumption $\bar{P}_\ell$ \eqref{eqn:Pell_bar} at time $t_{k+\ell}=t_{k}+\ell \Delta \tau$. To estimate the power consumption profile in the receding horizon window $\tau(t_k)$ we use $L$ instances of the dynamic average consensus algorithm, one for each $t_{k+\ell}$. Due to space limitations, we do not discuss here the details of dynamic consensus algorithms. The reader is referred to \cite{SolmazTutorial} for a comprehensive treatment of the topic. The local interaction protocol of a dynamic consensus algorithm can stated as
\begin{equation}\small
\bar{P}^i_{\ell}(t+dt)=\mathrm{D\_Consensus\_Update}(\bar{P}^j_{\ell}(t)|_{j\in\mathbfcal{N}^i},\bar{P}^i_{\ell}(t),p^i_\ell(t_k)).
\label{eqn:dyn_consensus}
\end{equation}
where ${dt}$ is time required to execute one iteration of the dynamic consensus protocol \eqref{eqn:dyn_consensus} according to the capabilities of the network. Local auxiliary variables of dynamic consensus algorithms are omitted here. Each agent updates its local estimation $\bar{P}^i_\ell$ through local interactions among its neighborhood $\mathbfcal{N}^i$, by exchanging their local estimation $\bar{P}^j_\ell$. Thus, preserving the agents' privacy since their own schedule plans $p^j_\ell(t_k)$ are note delivered to others. Due to the space limitation, details and notation of dynamic consensus algorithms are omitted. 


We point out that among the many dynamic consensus algorithms available in the literature, see e.g. \cite{spanos2005dynamic}-\cite{Franceschelli201969,Montijano20143131,Freeman2015},
for we decided to adopt the solution proposed in \cite{Franceschelli201969} for the following main reasons:
\begin{itemize}
    \item it can be easily tuned to achieve a desired trade-off between steady error and maximum tracking error;
    \item it is robust to re-initialization due to changes in the network topology or size;
    \item it  can be implemented with randomized asynchronous state updates, thus it is resilient against communication failures or agent logout during the algorithm execution.
\end{itemize}

At this point, we recall that  the performances of a dynamic consensus algorithm can be evaluated in terms of: a) the maximum tracking error of the estimated average of the time-varying reference signals
\begin{equation}\small
e_{\mathrm{track}}(t)=\max_{i\in \mathbfcal{V}} \left|\bar{P}^i_{\ell}(t)-\bar{P}_{\ell}(t)\right|=\max_{i\in \mathbfcal{V}}|\xi^i_{\ell}|,
\label{eqn:e_tracking}
\end{equation}
b) the steady-state error for constant references
\begin{equation}\small
e_{\mathrm{steady}}=\lim_{t\rightarrow \infty} \max_{i\in \mathbfcal{V}} \left|\bar{P}^i_{\ell}(t)-\bar{P}_{\ell}\right|,
\label{eqn:e_steady}
\end{equation}
and c) the convergence rate, which dictates how many iterations of the local state updates are required to achieve the steady-state error performance. Since in the proposed application the timescale $d t$ of the iterations of the dynamic consensus is in the range of the milliseconds, the timescale $\Delta t$ of the asychronous local optimizations executed by the SPS is in the range of the seconds, and a time-slot $\Delta \tau $ in the receding-horizon time window is in the range of minutes, most dynamic consensus algorithm in the literature are able to achieve their steady-state performance for all practical purposes. 

In our framework, to comply with typical assumptions of the dynamic consensus literature, we assume the connectedness of graph $\mathbfcal{G}$. Furthermore we assume that the maximum tracking error $\max_{i\in \mathbfcal{V}}|\xi^i_{\ell}|\leq \xi$ is bounded at all times.

\section{TCL Cooperation Control Protocol}\label{subsect:heuristic}

{
We now present the TCL Cooperation Control Protocol, detailed in Algorithm \ref{mainalgo}, which enables cooperation among the TCL by minimizing the global objective function \eqref{eqn:globalObjective} under the integer linear local constraints $\chi^i(t_k)$ derived in subsection~\ref{subsect:mld}. These constraints uniquely determine $u^i_{\ell}$ for $\ell=1,\ldots,L$ as function of the only decision variable, i.e., the SPS state $s^i_{\ell}$, based on  the current estimation of the virtual temperature $y^i_{\ell}$ obtained by the local observer, the current ON/OFF discrete state of the TCL, the system parameters identified in subsection \ref{subsubsect:Sys_Id}, the MLD modeling of the TCL plus SPS hybrid system and the average TCL power consumption profile of the network predicted via dynamic consensus. The considered global objective function is thus 
\begin{equation}\label{globalObjective}
\min_{s^i\in \chi^i(t_k),  i \in \mathbfcal{V}} J(t_k)
=\frac{1}{L}\sum_{\ell=1}^{L} \left(\sum_{i\in \mathbfcal{V}} \mathrm{p}_i \cdot u^i_{\ell}(t_k)\right)^2.
\end{equation}
}

\begin{algorithm}[ht!]
  \begin{algorithmic}[0]\small
  \Statex \textbf{ Algorithm's Parameters}:

    { $dt$: Maximum execution time of a dynamic consensus update;}
    
    { $\Delta t$: Maximum execution time of the local optimization;}

    { $\Delta \tau$: Time length of an optimization time-slot};

    {$L$: Number of time slots of the time window horizon;}
  
    {$\tau(t_k)=L\cdot \Delta \tau$: Optimization's receding horizon time window;}
  
    
    $\xi$: maximum tracking error of dynamic consensus algorithm;
    
    $\epsilon$: small positive constant;
    
    $\mu_i$: probability of execution of a local optimization;
  
   \Statex \textbf{ Algorithm Inputs}:
   
    { $y^i(t_k)$: Estimated virtual temperature;}
    
    { $\mathrm{Obs\_State}(t_k)=\lbrace{\mathrm{UR}, \mathrm{RE}}\rbrace$: Observer discrete state;}
    
    \Statex \textbf{Algorithm Outputs}: 
    
    {  $\mathbf{s}^i(t_k)=[s_1^i\cdots s_L^i]^\intercal\in\lbrace 0,1 \rbrace^L$: Scheduling plan of SPS $i$;}
    
    { $\mathbf{u}^i(t_k)=[u_1^i \cdots u_L^i]^\intercal\in\lbrace 0,1 \rbrace^L$: Predicted scheduling of power consumption of TCL $i$;}

    { $\mathbf{p}^i(t_k)=\mathrm{p}^i\cdot\mathbf{u}^i(t_k)$: Expected power consumption of TCL $i$;}


\\
    \Statex \textbf{Initialize counter}: $k=0$;

    \Statex \textbf{Execute in parallel the next tasks:}
\\
    \Statex $\bullet$ \textbf{Task a.} Every $dt$ seconds:

    \Statex \quad $1.$ Gather $\bar{P}_{\ell}^j$ for $\ell=1,\ldots,L$, from neighbors $j \in \mathbfcal{N}^i(t_k)$;
    \Statex \quad $2.$ Update state variables  $\bar{P}_{\ell}^i$ for $\ell=1,\ldots,L$, according to the \textbf{dynamic consensus algorithm} in \eqref{eqn:dyn_consensus};
    \\
   \Statex $\bullet$  \textbf{Task b.} Every $\Delta t$ seconds:
   \Statex \quad $1.$ Measure power consumption;
   \Statex \quad $2.$ Update the state of the virtual temperature observer and collect the virtual temperature $y_i(t_k)$;
   \Statex \quad $3.$ \textbf{If} $\mathrm{Obs\_State}(t_k)=\lbrace{\mathrm{RE}}\rbrace$, i.e., the local observer is in state "Reliable" \textbf{then}, with probability $\mu_i$ update the ON/OFF scheduling according to an approximate solution of the next problem by a time-constrained ($\Delta t$) heuristic:
\begin{equation}\small  \label{localOptproblem}
[\mathbf{s}^{i,\star},\mathbf{u}^{i,\star}]=\argmin_{\mathbf{s}^i\in \mathcal{\chi}(t_k)} \quad J^{i}(t_k)= \displaystyle \sum_{\ell=1}^L {P}_{\ell}^i p_i u^i_{\ell},
\end{equation}
\begin{equation}\small  \label{localoptimalsol}
J^{i,\star}= \displaystyle \sum_{\ell=1}^L {P}_{\ell}^i p_i u^{i,\star}_{\ell}.
\end{equation}
   

  \Statex \quad \quad \textbf{If} a solution $[\mathbf{s}^{i,\star},\mathbf{u}^{i,\star}]$ is found within $\Delta t$ seconds and
\begin{equation}\tiny
|J^{i}(t_k)-J^{i,\star}|=\gamma_i \geq \frac{\xi p_i}{L}\sum_{\ell=1}^{L} | u^{i,\star}_{\ell}-u^{i}_{\ell}|+\varepsilon,
  \nonumber
\end{equation}


  \Statex \quad \quad \textbf{then} set $\mathbf{s}^i(t_k):=\mathbf{s}^{i,\star}$, $\mathbf{u}^{i}(t_k):=\mathbf{u}^{i,\star}$
  \Statex \quad \quad \textbf{endif}
  \Statex \quad \textbf{else} $\mathbf{s}^i(t_k):=\mathbf{s}^i(t_{k})$, $\mathbf{u}^{i}(t_k):=\mathbf{u}^{i}(t_k)$, i.e., do nothing.
  \Statex \quad \textbf{Endif}

  { \Statex $\bullet$  \textbf{Task c.} Every $\Delta \tau$ seconds:}
  
  { \Statex \quad $1.$ Set the current SPS state equal to: $$\small s^i(t_{k+1}):=s^i_{1}(t_k);$$}
  {\Statex \quad $2.$ Shift the receding horizon time window by $\Delta \tau$:
    $$\small\mathbf{u}^i(t_{k+1}):=\begin{bmatrix}u^i_{2}(t_k) & \cdots & u^i_\ell(t_k) & \cdots & u^i_{L-1}(t_k) & 1\end{bmatrix}^\intercal;$$}
    { \Statex \quad $3.$ Shift SPS scheduling by $\Delta \tau$:
    $$\small\mathbf{s}^i(t_{k+1}):=\begin{bmatrix}s^i_{2}(t_k) & \cdots & s^i_\ell(t_k) & \cdots & s^i_{L-1}(t_k) & 1\end{bmatrix}^\intercal;$$}
    \Statex Let $k:=k+1$
    \Statex \textbf{Endif}
    \end{algorithmic}
  \caption{TCL Cooperation Protocol}
  \label{mainalgo}
\end{algorithm}

The ``TCL Cooperation Protocol'', consists of a local state update rule executed by each agent indefinitely. Each agent owns a local prediction of the future average TCL power consumption of the network over the horizon $\mathcal{\tau}(t_k)$. To update this prediction each agent executes the multi-stage dynamic consensus algorithm proposed in~\cite{FraGas2016,Franceschelli201969}. In particular, at each iteration each agent attempts to minimize, with probability $\mu$, a local objective function which consists in the agents ON/OFF SPS scheduling weighted by the predicted average TCL power consumption of the network over the receding horizon time-window, subject to the local MLD constraints $\mathcal{\chi}^i(t_k)$ illustrated in subsection~\ref{subsect:mld}. 
Notice that, the local constraints of each agent $\mathcal{\chi}^i(t_k)$ are time-varying because they depend upon the current state of the TCL at the time the optimization takes place.


Although the optimization problem in \eqref{localOptproblem} is in general NP-hard, as it involves mixed integer linear programming, in our setting for each agent~$i$, the number of variables to be optimized is relatively small as only local constraints over a short time horizon $\mathcal{\tau}(t_k)$ are involved.  For example, about $20$/$60$ steps into the future may account for $30$ minutes to two-three hours of operations, depending on the tuning of the algorithm. Thus, the complexity of the proposed method does not increase with the size of the network, but only with respect to the size of the time horizon $L$. This can be shown by noticing that the local optimization executed by each agent involves only its own state and its own prediction of the network average power consumption while cooperation with other agents is achieved only through the execution of the dynamic consensus protocol which is designed for large scale networks.

Furthermore, it should be noticed that approximate solutions to local optimization problems are sufficient to execute the heuristic. In particular, we exploit a standard branch and bound solver with a limited maximum execution time. Finally, if an approximate solution does not improve on the current scheduling of operation of the generic agent, it is simply discarded and the existing scheduling is kept out until a better one is found in future iterations. In section \ref{sect:simulation} we provide a detailed discussion of numerical simulations and experimental tests carried out for our case study. 

\subsection{Convergence analysis}
Next, we prove that a feasible solution to the local optimization problem in \eqref{localOptproblem} always exists. 

\begin{propo} [Local problem feasibility]
There exist at least one feasible solution to the local optimization Problem \eqref{localOptproblem} for all $t_k\geq 0$.
The solution is $s^i(t_{k+\ell})=1$, for $\ell=1,\ldots,L,$ which corresponds to an SPS always ON.








\end{propo}

\begin{proof}
We now show that by construction the list of constraints included in set $\chi^i(t_k)$ as in Definition \ref{constraintDefinition} is always satisfied by a solution with $s^i(t_{k+\ell})=1$ for $\ell=1,\ldots,L,$.

First, the equality in constraint \eqref{eqn:y_update_big} represents the linear dynamics of the TCL virtual temperature, as such for any choice of control $\mathbf{u}^i(t_k)$ and initial virtual temperature $y^i_{1}$ the predicted virtual temperature vector $\mathbf{y}^i(t_k)$ exists and is uniquely defined.

Constraints \eqref{eqn:X_relay:eq1_X1} and \eqref{eqn:X_relay:eq3_X1} define the value of the boolean auxiliary variables $\delta^i_{1,\ell}$ and $\delta^i_{2,\ell}$ which represent the conditions $\delta^i_{1,\ell}=1 \biff y^i_{\ell}\leq y^i_{min}$ and  $\delta^i_{2,\ell}=1 \biff y^i_{\ell}\leq y^i_{max}$. Therefore, given vector $\mathbf{y}^i(t_k)$ it follows that $\delta^i_{1,\ell}$ and $\delta^i_{2,\ell}$ are uniquely defined for $\ell=1,\ldots,L$.

Constraint \eqref{eq19} is trivially satisfied by $s^i_{\ell}=1$.

If we substitute for $s^i_{\ell}=1$  in constraints \eqref{eq21}, \eqref{eq20}, \eqref{eq17}, \eqref{eq18} and  \eqref{eq22} we obtain the thermostatic control logic in \eqref{eqn:u_hyst} modeled via MLD, thus for all possible values of virtual temperature and TCL state there is a corresponding (always feasible) control input.

Finally, constraints  \eqref{eq22} and \eqref{eqn:s_limit_off} are trivially satisfied by $s^i(t_{k+\ell})=1$ for $\ell=1,\ldots,L$.  \hfill $\blacksquare$
\end{proof}

Next, we present a characterization of the convergence properties of Algorithm \ref{mainalgo}, i.e., if at each iteration the dynamic consensus algorithm has bounded tracking error than the proposed TCL cooperation protocol optimizes online the global objective function \eqref{globalObjective}.

Let $J^+(t_k)$ be the value of the global objective $J(t_k)$ \eqref{globalObjective} after one agent updates its own local control action $\mathbf{s}^i(t_{k})$ during the execution of Task (b) of Algorithm~\ref{mainalgo}.

\begin{theorem}\label{theo:online_opt}[Online Optimization]
Consider a network of agents that executes Algorithm~\ref{mainalgo}.  If the dynamic consensus algorithm executed as part of Algorithm~\ref{mainalgo} in Task (a) has maximum tracking error less than $\xi$ and it holds  
$$\small
|J^{i}(t_k)-J^{i,\star}|=\gamma_i \geq \frac{\xi p_i}{L}\sum_{\ell=1}^{L} |u^{i,\star}_{\ell}-u^i_{\ell}|+\varepsilon,
$$
where  $J^{i,\star}$ is the approximation of the optimal local solution computed by agent $i$ at step $3$ of Algorithm~\ref{mainalgo}, then the global objective value decreases and it holds
\begin{equation}\small
J^+(t_k) \leq J(t_k)-n\varepsilon.
\end{equation}
where $\varepsilon$ is a small positive constant and $n$ is the number of agents.
\end{theorem}
\begin{proof}
The global objective to be optimized online is
\begin{equation}\label{equazione01}\small
J(t_k)=\frac{1}{L} \sum_{\ell=1}^{L} \left(\sum_{i\in \mathbfcal{V}} p_i u^i_{\ell}\right)^2.
\end{equation}
Now, the average power consumption estimated by each agent with the dynamic consensus algorithm has, in general, a time-varying estimation error $\xi^i_{\ell}$ with respect to the real average power consumption $\bar{P}_{\ell}$ such that
\begin{equation}\label{equazione06}\small
\bar{P}_{\ell}=\frac{\sum_{i\in \mathbfcal{V}} p_i u^i_{\ell}}{n}=\bar{P}^i_{\ell}+\xi^i_{\ell}.
\end{equation}
Thus, we can rewrite \eqref{equazione01} as
\begin{equation}\label{equazione02}\small
\begin{array}{ll}
J(t_k)& =\frac{1}{L} \displaystyle \sum_{\ell=1}^{L}\left(\sum_{i\in \mathbfcal{V}} p_i u^i_{\ell}\right) \left( \sum_{i\in \mathbfcal{V}} p_i u^i_{\ell}\right) \\
&= \frac{1}{L} \displaystyle \sum_{\ell=1}^{L} n \bar{P}_{\ell} \left(\sum_{i\in \mathbfcal{V}} p_i u^i_{\ell}\right)\\
& = \frac{1}{L} \displaystyle  \sum_{\ell=1}^{L} \sum_{i\in \mathbfcal{V}}n \bar{P}_{\ell}  p_i u^i_{\ell} \\
& = \displaystyle n \sum_{i\in \mathbfcal{V}} \frac{1}{L} \sum_{\ell=1}^{L} \left(\bar{P}^i_{\ell}+\xi^i_{\ell}  \right)  p_i u^i_{\ell}(t_k). \\
\end{array}
\end{equation}
Now, we first notice that
\begin{equation}\label{equazione03}\small
J(t_k)=\displaystyle n \sum_{i\in \mathbfcal{V}} J^i(t_{k}),
\end{equation}
where
\begin{equation}\label{equazione04}\small
J^i(t_{k})=\frac{1}{L}\sum_{\ell=1}^{L} \left(\bar{P}^i_{\ell}+\xi^i_{\ell}  \right)  \mathrm{p}_i u^i_{\ell}.
\end{equation}
Therefore, if exact knowledge of $\bar{P}_{\ell}=\bar{P}^i_{\ell}+\xi^i_{\ell}$ were available to each agent, we could guarantee the optimization of the global objective function over the non-convex set of constraints (up to a local minimum) by updating the local ON/OFF scheduling as
\begin{equation}\small
[\mathbf{s}^{i,\star},\mathbf{u}^{i,\star}]=\argmin_{\mathbf{s}^i\in \mathcal{\chi}(t_k)}\quad \displaystyle \frac{1}{L} \sum_{\ell=1}^L \bar{P}_{\ell} \mathrm{p}^i u^i_{\ell}.
\end{equation}

Instead, since the agents do not have access to global information regarding the network, i.e., they do not have access to the updated values of $\bar{P}_{\ell}$ during the iterations, an estimation of $\bar{P}_{\ell}$ is employed. 

Since we consider a dynamic consensus process which has bounded tracking error less than $\xi$, it holds
\begin{equation}\label{equazione7}\small
\max_{i\in\mathbfcal{V}} |\bar{P}^i_{\ell}-\bar{P}_\ell|\leq \xi, \quad \forall i\in \mathbfcal{V}, \quad \ell=1,\ldots,L. 
\end{equation}
During the execution of task (b) of Algorithm~\ref{mainalgo}, each agent optimizes its local objective function based on its local estimation of the predicted average power consumption in the network $\bar{P}^i_{\ell}$, i.e.,
\begin{equation}\label{equazione10}\small
\tilde{J}^i(t_k)=\frac{1}{L} \sum_{\ell=1}^{L} \bar{P}^i_{\ell} \mathrm{p}^i u^i_{\ell}=J^i(t_k)+\frac{1}{L} \sum_{\ell=1}^{L} \xi^i_{\ell} \mathrm{p}^i u^i_{\ell}.
\end{equation}
Thus, let
\begin{equation}\small
[\mathbf{s}^{i,\star},\mathbf{u}^{i,\star}]=\argmin_{\mathbf{s}^i\in \mathcal{\chi}_i(t_k)} \tilde{J}^i(t_k)=\argmin_{\mathbf{s}^i\in \mathcal{\chi}_i(t_k)} \displaystyle \frac{1}{L} \sum_{\ell=1}^L \bar{P}_{\ell}^i \mathrm{p}^i u^i_{\ell},
\end{equation}
and
$$
\tilde{J}^{i,\star}=\frac{1}{L} \sum_{\ell=1}^{L} \bar{P}^i_{\ell} \mathrm{p}^i u^{i,\star}_{\ell}.
$$
%
%
%
Let us now denote with $\gamma_i$ the computed decrement of the local objective function affected by estimation errors
\begin{equation}\label{equazione05}\small
\tilde{J}^{i,\star}-\tilde{J}^i(t_k)= -\gamma_i.
\end{equation}
We now compute a sufficient value of $\gamma_i$ which guarantees the optimization of the global objective function despite persistent estimation errors. 

An actual decremet of the local objective function after one agent executes task (b) of Algorithm \ref{mainalgo} is obtained if 

\begin{equation}\label{decrement}
J^{i,\star}(t_k)-J^i(t_k)< -\varepsilon.
\end{equation}
By rewriting  \eqref{equazione05} and substituting $\tilde{J}^i(t_k)$ by exploiting \eqref{equazione10}, it holds
\begin{equation}\label{decrement2}
J^{i,\star}(t_k)-J^i(t_k)
 = -\gamma_i+\frac{1}{L} \sum_{\ell=1}^{L} \xi^i_{\ell} \mathrm{p}_i (u^{i,\star}_{\ell}-u^i_{\ell}).
\end{equation}
%
%
%
%
%
%

Then, by putting together the inequalities in~\eqref{decrement} and in~\eqref{decrement2},  it follows that to ensure~\eqref{decrement}, it suffices that
\begin{equation}
- \gamma_i  + \frac{1}{L} \sum_{\ell=1}^{L} \xi^i_{\ell} \mathrm{p}_i (u^{i,\star}_{\ell}-u^i_{\ell}) < -\varepsilon
\end{equation}

Thus, a decrement occurs if
\begin{equation}\label{decrement3}
\gamma_i >  \frac{1}{L} \sum_{\ell=1}^{L} \xi^i_{\ell} \mathrm{p}_i (u^{i,\star}_{\ell}-u^i_{\ell}) +\varepsilon
\end{equation}

%
%
%
%
By considering an upper bound to the estimation error $\xi^i_{\ell}$ by exploiting \eqref{equazione7}, it holds $|\xi^i_{\ell}|\leq \xi$ for all $\ell=1,\ldots,L$, and for all $i\in \mathbfcal{V}$. Thus, we can rewrite \eqref{decrement3} as the next inequality
\begin{equation}  \label{equazione16}\small
\gamma_i>\frac{\xi \mathrm{p}_i}{L} \sum_{\ell=1}^{L} |u^{i,\star}_{\ell}-u^i_{\ell}|+\varepsilon.
\end{equation}
%
%
%

Therefore, if \eqref{equazione16} holds, the value of the global objective function $J^+(t_k)$ after one agent executes task (b) is
\begin{equation}\small
\begin{array}{ll}
J^{+}(t_k)&= \displaystyle n \left(\left(\sum_{j\in \mathbfcal{V}\setminus{i} } J^j(t_{k})\right)+J^{i,\star}\right)\\
\end{array}
\end{equation}
and decreases, with respect to its value before the local update, by at least 
\begin{equation}\label{equazione18}\small
\begin{array}{ll}
J^{+}(t_k)&< \displaystyle n \left(\left(\sum_{j\in \mathbfcal{V}\setminus{i} } J^j(t_{k})\right)+J^i(t_k)-\varepsilon\right)\\
& =J(t_k)-n\varepsilon
\end{array}
\end{equation}
thus proving the statement of this theorem. \hfill $\blacksquare$
\end{proof}

\section{Description of the testbed and experimental validation}\label{sect:simulation}

Here, we firstly describe the low-cost experimental test-bed
developed to validate the proposed multi-agent DSM control architecture. Then, we propose an experimental validation of the proposed method involving domestic TCL appliances such as water heaters and radiators located in private homes of a set of volunteers participating in the pilot. Furthermore, to validate the effectiveness of the proposed framework on a larger scale while keeping the experimental complexity manageable, we also decided to introduce  in the network of cooperating devices some virtual TCLs, i.e., TCLs that are numerically simulated. \\


\subsection{Description of the CoNetDomeSys experimental testbed}

A significant contribution presented in this paper is the "CoNetDomeSys" testbed, short for ``Cooperative Network of Domestic Systems'', a low cost IoT-oriented experimental demonstrator for fast prototyping and testing of DSM algorithms on large populations of domestic appliances controlled and monitored by smart power sockets. 

The testbed is designed around off-the-shelf low-cost hardware components. In particular, the core component of the testbed is the WeMo$^\circledR$ Insight Switch \cite{wemo} smart power socket. The choice of this particular SPS was dictated by availability of Open-APIs to integrate purpose-built software, thus enabling remote monitoring and control. A WeMo$^\circledR$ Insight Switch is provided with 
\begin{itemize}
    \item a power consumption sensing unit with a resolution of $1\mathrm{mW}$ and tested maximum  sampling frequency of $1\mathrm{Hz}$;
    \item a switch to remotely power ON and OFF the appliance plugged in;
    \item a microcontroller and WiFi communication capability. 
\end{itemize}

A number of SPS were delivered to voluteers in the city of Cagliari who agreed to participate to the experimental validation of the proposed architecture. In each domestic environment a Raspberry Pi Zero W \cite{raspPi0w} was installed, connected to the same WiFi LAN of the SPS. Each SPS has been connected to a domestic TCL. 

The software developed for our testbed consists in two Java\texttrademark applications: one installed in each Raspberry Pi Zero W, that implements a client-server communication and control infrastructure over the Telnet protocol for domestic appliances; and one installed in a workstation. 

The application running on the Raspberries manages the local monitoring and control of each SPS, i.e., the power absorbed by the load $p^i(t_k)$ with an associated time-stamp and the discrete state of the SPS $s^i(t_k)$. Furthermore, it forwards measurements and receives control commands from a workstation with public IP address.

The application running on  the workstation collects data from the large population of SPSs, manages a database and sends actuation commands. A Matlab interfaces has been integrated with the software for the fast prototyping of distributed coordination algorithms. To each agent is associated a real SPS, tested communication delays are below $1$ sec.

A network topology is assigned to the agents and each can only share information anonymously with its neighbors. The topology and size of the network is unknown to the agents. 

The  workstation used for the experimental validation of our approach is a Dell Precision t5810 workstation equipped with an Intel$^\circledR$ Xeon$^\circledR$ E5-1620 v3 (10M Cache, 3.50 GHz), 64GB of RAM, Windows 10 Pro. 



The processing carried out by the testbed is centralized and takes place in a single workstation. Therefore, the testbed validates experimentaly the sensing and actuation of the SPS and TCLs while to enable fast algorithm protoyping and testing in the MATLAB$^\circledR$ environment the processing and communications among agents are simulated via software.


Technical details concerning the hardware and software architecture for implementing the proposed multi-agent control architecture on SPSs can be found in a patent filed to the Italian Patent and Trademark Office \cite{patentCoNetDomesys}.

\subsection{Parameters and setting of the experimental scenario}

\begin{table}[!tbp]
    \caption{Test parameters}
    \label{tab:TCLparam}
    \centering
    \begin{tabular}{|c|c|c|c|c|}
    \toprule
      $\#$ &\textbf{TCL}   & $\alpha^i~\mathrm{[s^{-1}]}$ & $\beta^i\cdot q^i/\alpha^i~\mathrm{[-]}$ & $\mathrm{p}^i~\mathrm{[kW]}$\\
      \midrule
      $12$ & Water Heater   &  $\left(0.3\div5.3\right)\times10^{-4}$ & $6.17\div48.25$& $1.2\div1.5$\\
       \midrule
      $3$ &  Radiator  &  $\left(1.5\div3.5\right)\times10^{-3}$ & $1.6\div5.5$& $1.2,2.0$\\
       \midrule
      $85$ & Simulated TCL& $\left(1.3\div2.4\right)\times10^{-4}$ & $21\div25$& $1.0\div2.0$\\
     \bottomrule
    \end{tabular}
\end{table}

{To validate the proposed  approach on the CoNetDomeSys testbed a small scale scenario involving domestic TCL appliances located in a set of private homes of volunteers has been considered. To limit the number of volunteers required to carry out the test and execute the TCL cooperation protocol on a network of at least $100$ agents we also considered a set of numerically simulated TCLs and SPS. In particular, in the present scenario we had access to 12 electric water heaters and 3 electric radiators in $10$ different locations in the city of Cagliari, Italy. 




\begin{figure}[!ht]
\centering
\includegraphics[width=9pc]{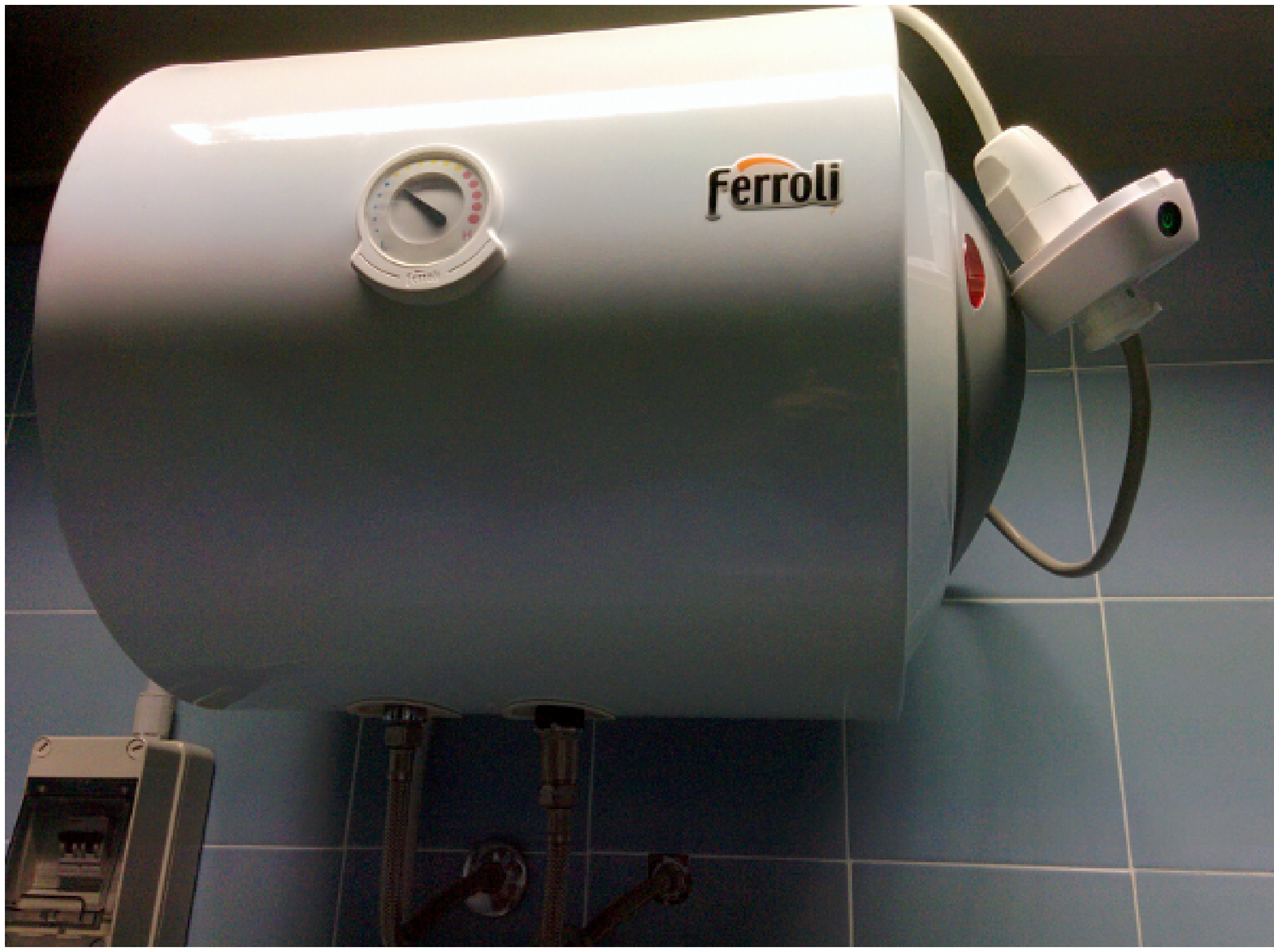}
\includegraphics[width=8.87pc]{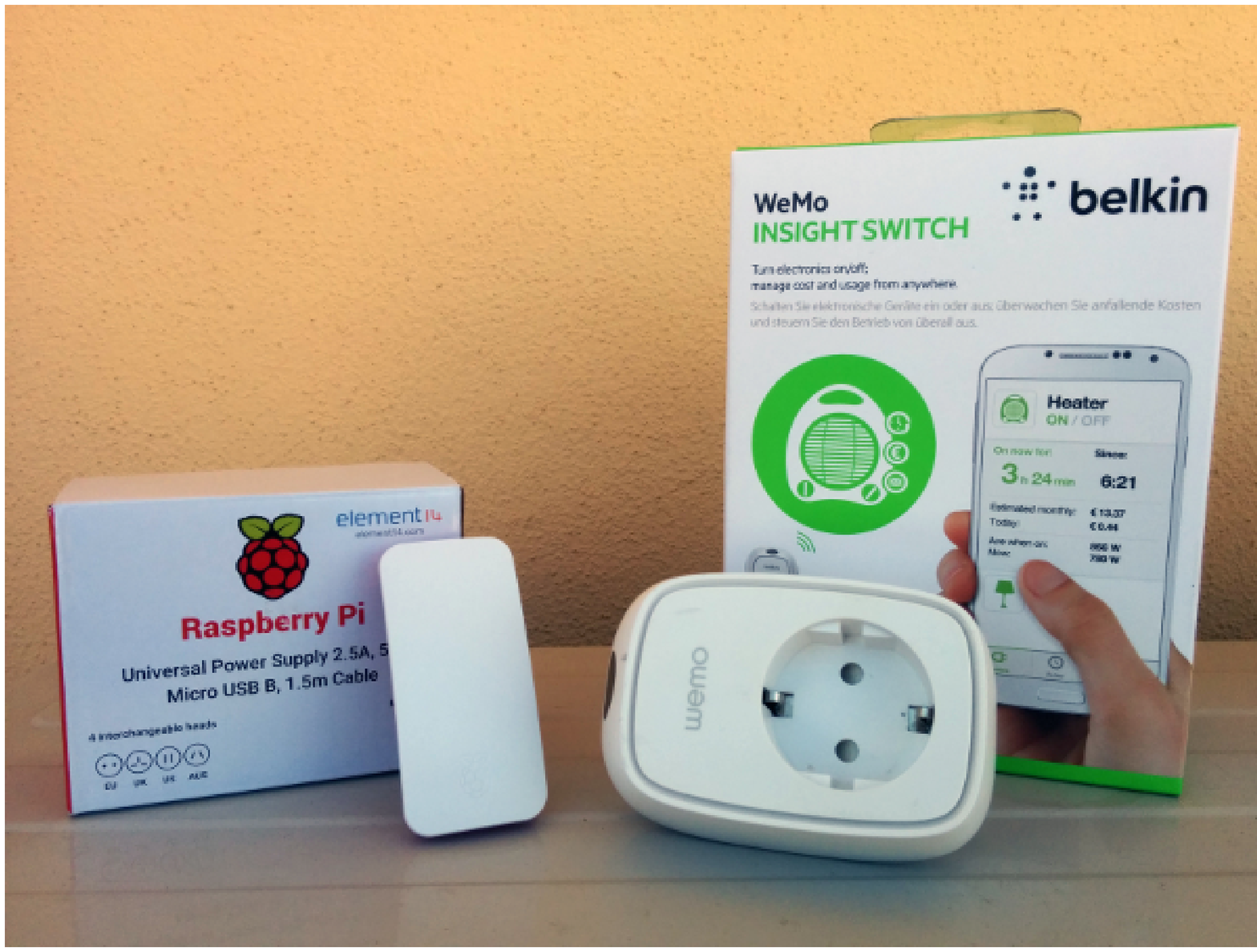}
\caption{Left: An electric water heater plugged into a smart power socket. Right: A Raspberry Pi Zero W and a WeMo$^\circledR$ Insight Switch.
}
\label{fig:testbed}
\end{figure}
\begin{figure}[!ht]
\centering
\includegraphics[width=21pc]{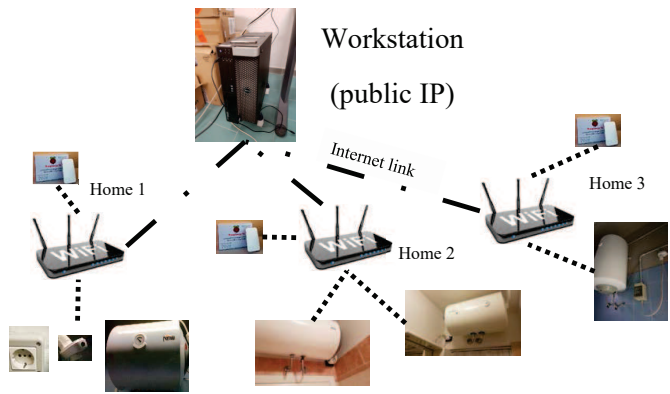}
\caption{Example of physical/communication topology of the CoNetDomeSys Testbed.}
\label{fig:testbed}
\end{figure}

We introduced $85$ numerically simulated SPS and TCLs with parameters identified from other real TCLs involved in the experiment, thus yielding a total population of \mbox{$n=100$} agents.
On the left side of Figure~\ref{fig:testbed} one of the electric water heater under test and Raspberry Pi Zero W and a WeMo$^\circledR$ Insight Switch are shown.} The parameters of each TCL were identified with the method described in subsection \ref{subsubsect:Sys_Id} and by choosing parameter ``$y^i_{\mathrm{min}}=0.5$'' for all agents.


\begin{figure}[!ht]
\centering
\includegraphics[width=21pc]{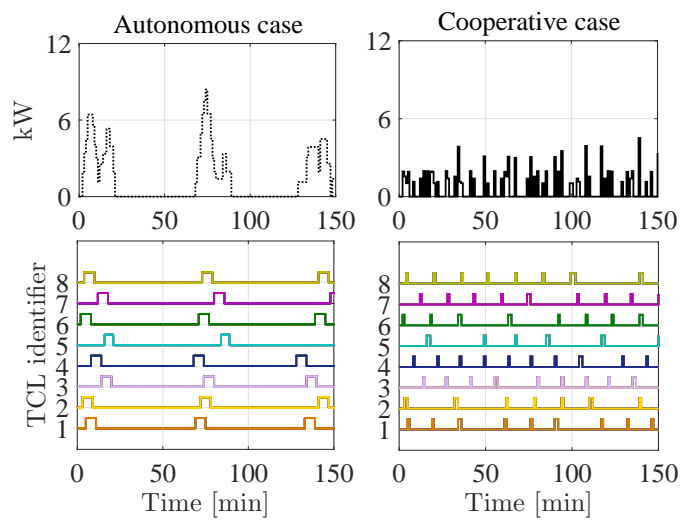}
\caption{Comparison between the power consumption of $8$ simulated TCLs in autonomous operation (autonomous case) and when executing the TCL cooperation protocol (cooperative case). The power consumption of each TCL is detailed with a different color.}
\label{fig:de-synchronization}
\end{figure}
    
In our scenario we considered a peer-to-peer network topology represented by a random undirected Erd\"os-R\'enyi graph $\G(\V,\E)$, with edge existence probability $3\log(n)/n$. Tests were carried out in real-time and in accordance with the notation in Algorithm~\ref{mainalgo} the time-related parameters were set as
\begin{itemize}
    \item Dynamic consensus (task a): $dt=10\mathrm{msec}$;
    \item Local optimization (task b): $\Delta t=1\mathrm{sec}$;
    \item SPS actuation time slots (task c):  $\Delta \tau=1\mathrm{min}$.
\end{itemize}

A receding horizon time window $\tau(t_k)$ of $40\mathrm{min}$ was chosen, thus $L=40$ time slots with length $\Delta \tau = 1 min$ each. The probability of executing the local optimization every $\Delta t$ on each device was set to $\mu^i=1/30$ for all agents. 
The expected number of local optimization rounds (see \eqref{localOptproblem}) executed by each agent in one time-slot $\Delta \tau$, is $\frac{\mu^i\Delta \tau}{\Delta t}=2$. Local optimization problems were solved by the MATLAB$^\circledR$ Mixed Integer Linear Programming solver ({\tt intlinprog}) with time limit of $5$ sec. The solver uses a time-constrained ``Branch and Bound'', in our scenario the average execution time of each optimization was about $0.03\mathrm{sec}$, exact optimal solutions were found within the time limit in the large majority of cases.

It should be noticed that, according to the update rule given at step 3 of Algorithm~\ref{mainalgo}, only solutions which improve the current local objective function by a minimum amount are exploited to update the ON/OFF scheduling of the SPS. Indeed, this ensures that the online optimization of the global objective function is carried out despite estimation errors.

In our test with $100$ agents we considered a time span of $550\mathrm{min}$ ($\approx$~9 hours). 

Numerically simulated TCLs and SPS were given random initial conditions with virtual temperature in the interval $[0.5,0.7]$. This ensured a scenario where about $\mathrm{90\%}$ of virtual TCLs would need to turn ON the heating within the first $25\mathrm{min}$ of the experimentation, thus leading to a scenario where the network of TCLs would tend to be synchronized and induce a significant peak of power consumption.


In Figure~\ref{fig:de-synchronization} the detail of the power consumption profile of $8$ simulated TCLs is shown. It can be seen that the discrete state of each TCL switches with increased frequency in the cooperative case.

\begin{figure}[!ht]
\centering
\includegraphics[width=21pc]{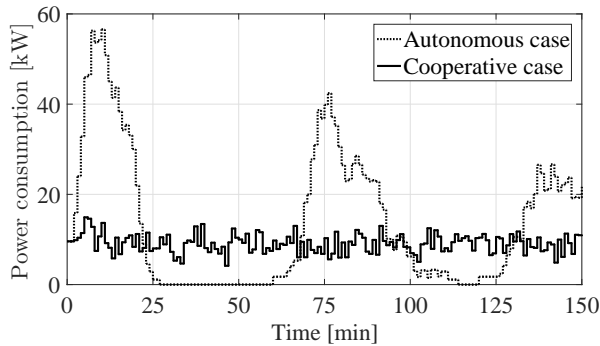}
\caption{Comparison between the power consumption of $100$ TCLs in autonomous operation (autonomous case) and when executing the TCL cooperation protocol (cooperative case) in the mixed scenario ($16$ real TCLs).}
\label{fig:de-synchronization1-mauro}
\end{figure}

\subsection{Experimental validation}\label{subsec:validation}

In Figure~\ref{fig:de-synchronization1-mauro} it is shown a comparison between the power consumption profile by the considered network of $100$ mixed TCLs ($15$ real and $85$ simulated TCLs) during autonomous operation (autonomous case) and when executing the TCL cooperation protocol (cooperative case). The numerically simulated TCLs have been set with the same initial conditions as for the autonomous case, the real TCL had different and arbitrary initial conditions. 
It is evident that the cooperation among TCLs promotes the desynchronization of their power consumption, thus reducing peak power demand.
%
%
\begin{figure}[!ht]
\includegraphics[width=21pc]{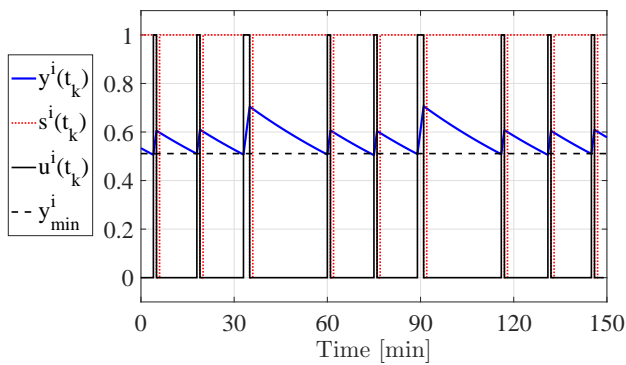}
\caption{Evolution of virtual temperature $y^i(t_k)$, SPS state $s^i(t_k)$, and TCL state $u^i(t_k)$ of one agent during the execution Algorithm~\ref{mainalgo}.}
\label{fig:virtual_TCL}
\end{figure}
In Figure~\ref{fig:virtual_TCL} it is shown a detail of the evolution of the discrete state of the SPS and TCL with the corresponding virtual temperature for one agent that is executing the TCL cooperation protocol. It can be seen that the SPS is most often in the ON state, thus there is no significant issue with excessive power switching that might otherwise shorten the lifespan of the device. Furthermore, it can be noticed that to modulate the state of the TCL while cooperating with others, the average temperature of the TCL seldom reaches the upper limit $T^i_{max}$ of the desired temperature range, as opposed to the standard thermostat control logic which reaches $T^i_{max}$ during every ON/OFF cycle, thus realizing some limited energy savings as a by product. 

\begin{figure}[!t]
\centering
\includegraphics[width=21pc]{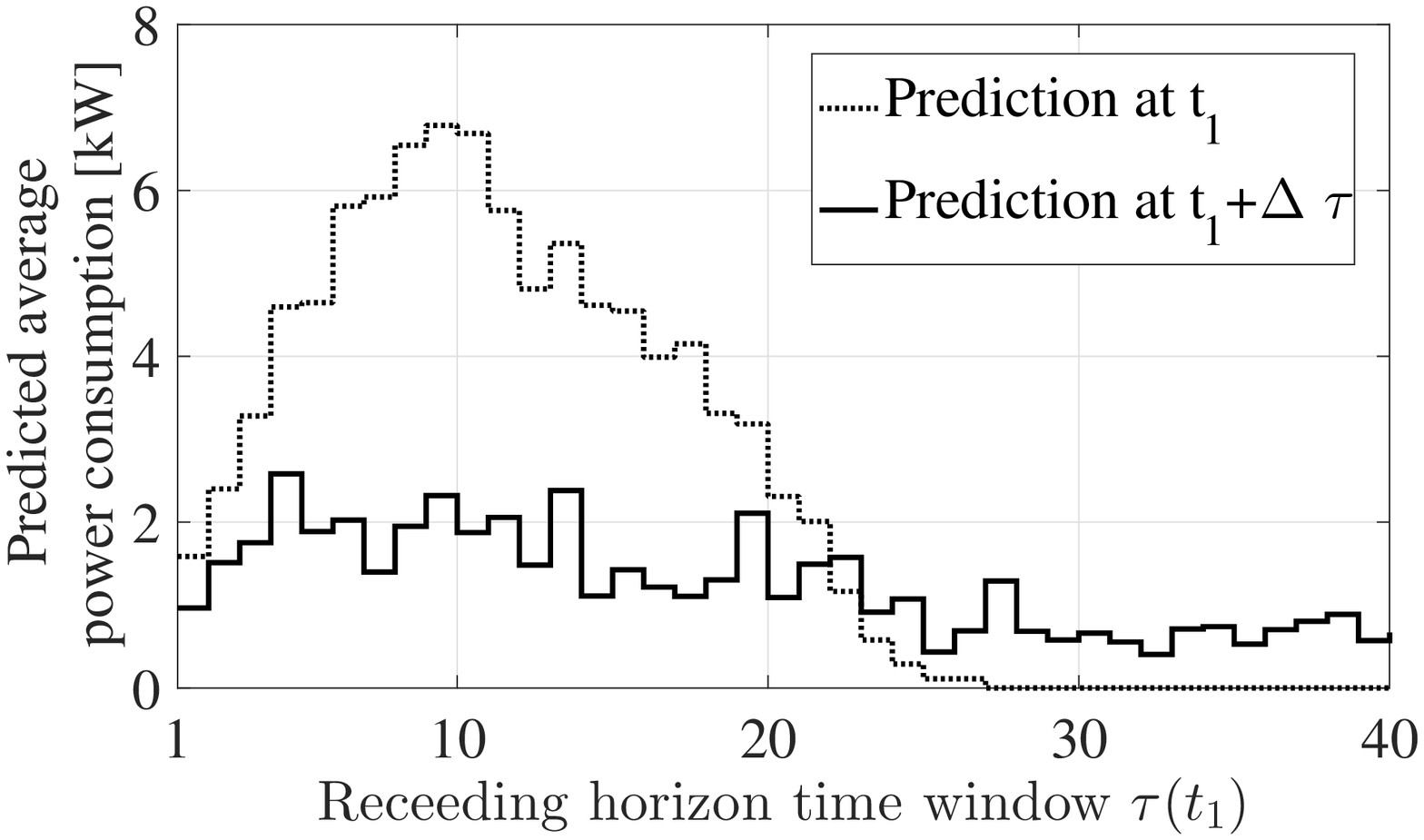}
\caption{Comparison of the predicted average power consumption  $\bar{P}^i_{\ell}$ by agent $i=1$, $\ell=1,\dots,40$, at time $t_1$ and time $t_1+\Delta \tau$.}
\label{fig:dyn_consensus}
\includegraphics[width=21pc]{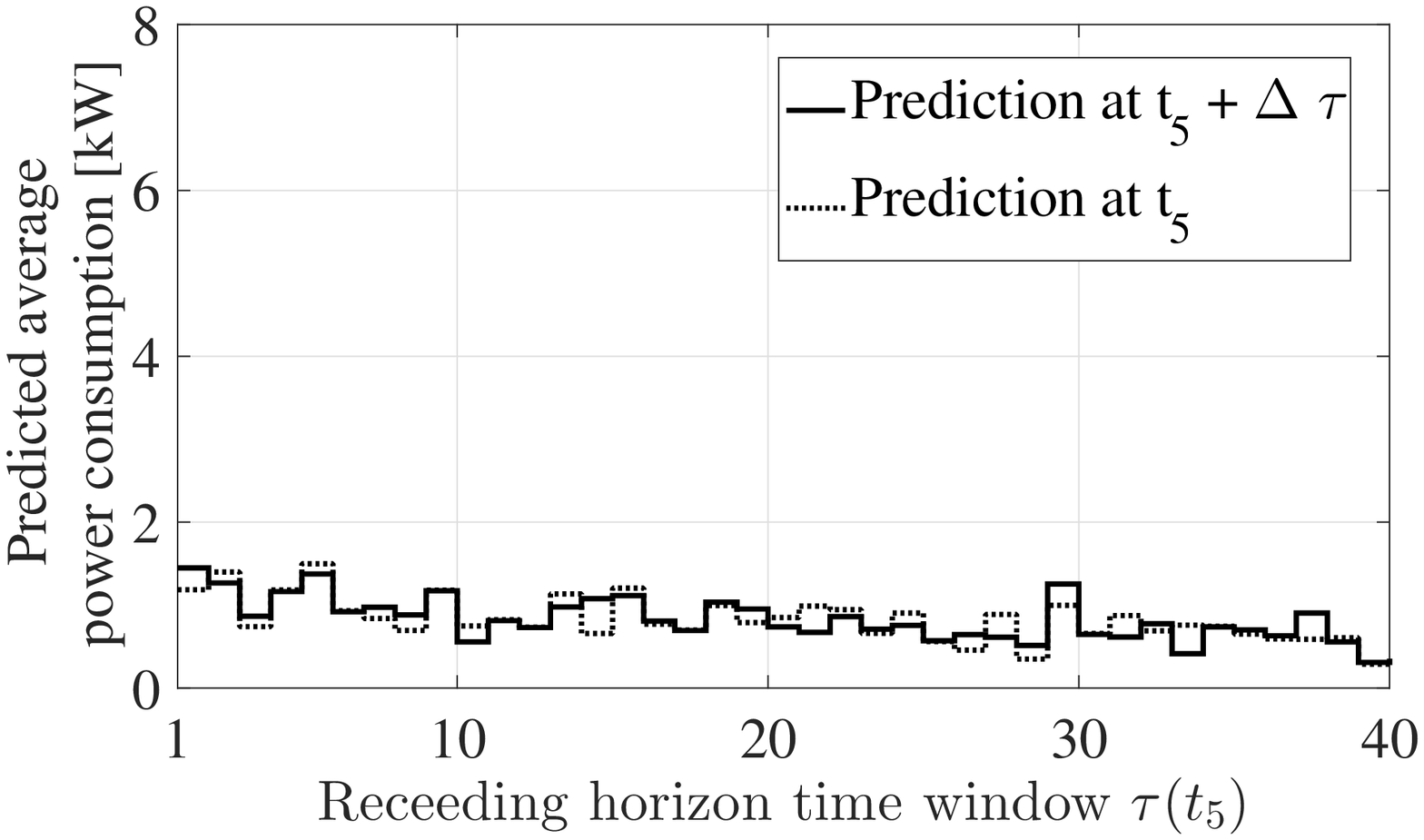}
\caption{Comparison of the predicted average power consumption  $\bar{P}^i_{\ell}(t_k)$ by agent $i=1$, $\ell=1,\dots,40$, at time $t_5$ and time $t_5+\Delta \tau$. }
\label{fig:dyn_consensus2}
\end{figure}

In Figure~\ref{fig:dyn_consensus} it is shown a comparison between the predicted average power consumption $\bar{P}^i_\ell$ in \eqref{eqn:dyn_consensus} for $\ell=1,\dots,40$ estimated with the dynamic consensus algorithm by agent $1$ at two different instants of time, evaluated at $t_1=1\mathrm{min}$ and at $t_1+\Delta \tau=2\mathrm{min}$. The same comparison is also evaluated at $t_5=5\mathrm{min}$ and at $t_5+\Delta \tau=6\mathrm{min}$ as shown in Figure~\ref{fig:dyn_consensus2}. It can be seen that major changes in the predicted power consumption by the network occur mostly during the transient behavior when the TCL cooperation protocol is initialized (Figure~\ref{fig:dyn_consensus}). After the time-window  $\tau(t_k)$ recedes by a few time-slots, the network reaches a steady state (Figure~\ref{fig:dyn_consensus2}) where only small changes occur to the predicted average power consumption and therefore small changes are triggered to the ON/OFF schedules of the SPS by local optimizations.

\begin{figure}[!h]
\centering
\includegraphics[width=21pc]{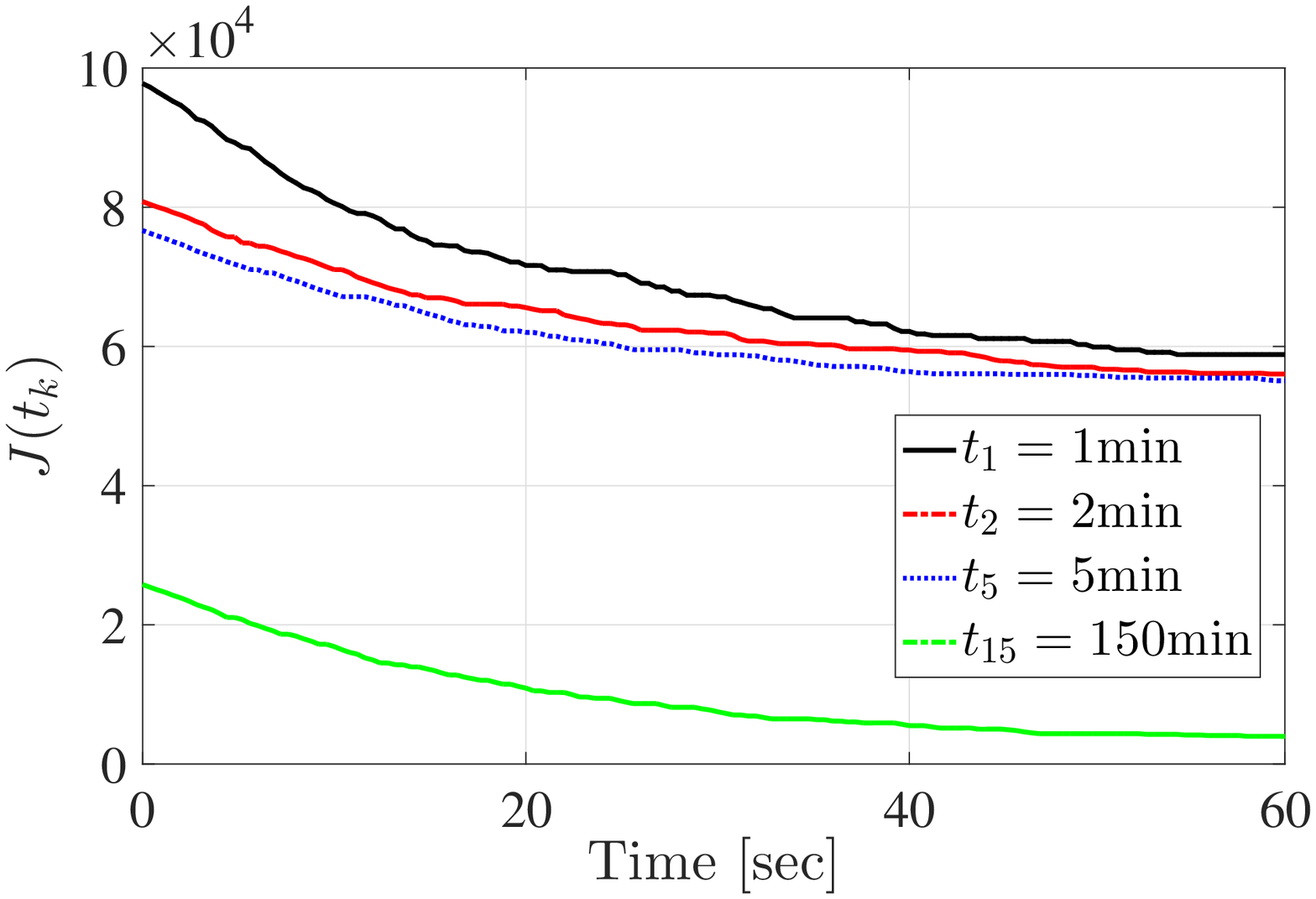}
\caption{Evolution of the global objective function during the execution of the TCL cooperation protocol at different intervals of time $t_k$.}
\label{fig:J_global_cost_decaying}
\vspace{0.5cm}
\includegraphics[width=21pc]{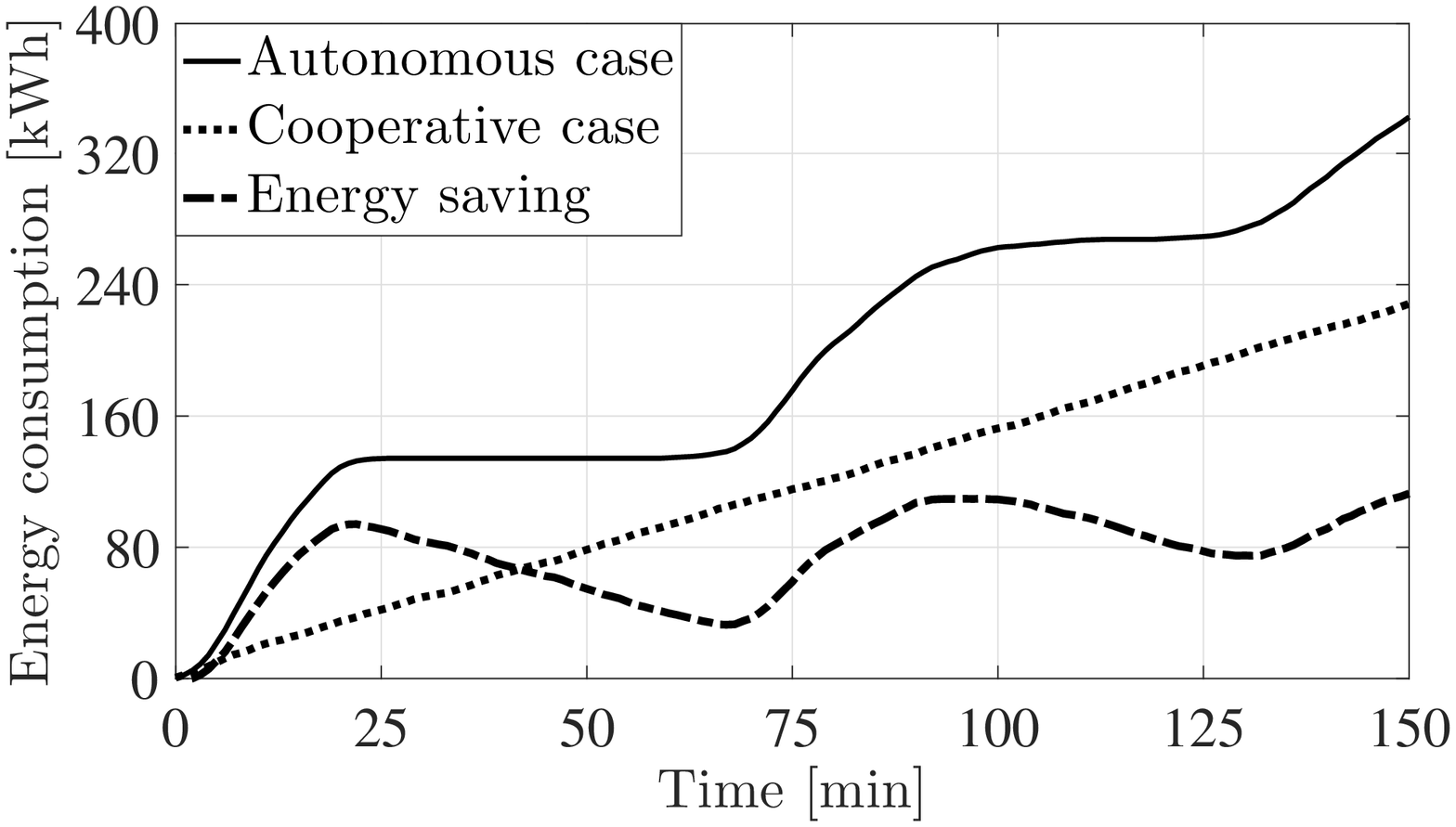}
\caption{
Comparison  between the energy consumption profiles, integrated within a $150$ min time window, with the TCLs in autonomous operation and during the execution of the TCL cooperation protocol.}
\label{fig:energy_saving}
\end{figure}


The experimental test shows that even a few iterations of task (b) the TCL cooperation protocol are sufficient to significantly improve the global objective, i.e., reduce peak demand and load variations. Since the approach is real-time and based on feedback, errors due to local TCL parameter uncertainty and estimation errors are mitigated and averaged within the whole network, thus providing robustness.

Figure~\ref{fig:J_global_cost_decaying} is a validation of the result in Theorem~\ref{theo:online_opt}, it shows the executing of the TCL cooperation protocol with randomized local optimizations, i.e., task (b), provides a decrement on the global objective $J(t_k)$ \eqref{eqn:globalObjective} despite estimation errors and real dynamical evolution of the TCLs.

Finally, in Figure~\ref{fig:energy_saving} it is shown a comparison  between the energy consumption profiles, integrated within a $150$ min time window, with the TCLs in autonomous operation and during the execution of the TCL cooperation protocol. It can be seen that as a by product of the proposed cooperation strategy, also total power consumption by the network is reduced, thus realizing energy savings. This occurs because the modulation of the ON/OFF state of the TCLs forces their temperature to be closer to the lower limit of the desired temperature range $[T^i_{min},T^i_{max}]$ thus reducing thermal losses with respect to the ambient temperature.{\hfill $\blacksquare$}

\begin{figure}[!ht]
\centering
\includegraphics[width=22pc]{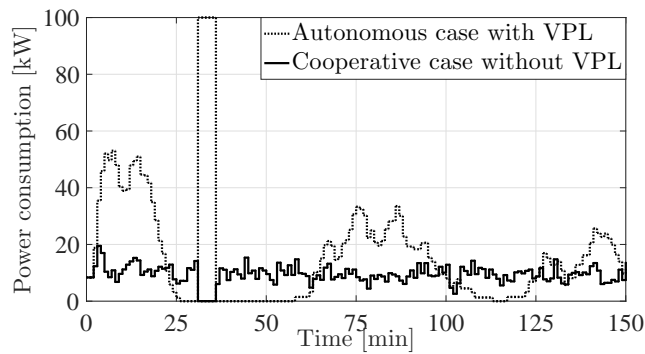}
\caption{Temporal  behaviour  of  the  actual  absorbed  power  of  the  network of TCL-plus-SPS systems without (left) and with (right) the proposed DSMscheme, in the presence of a VPL agent.}
\label{fig:vpl_test}
\end{figure}

\subsection{Experimental validation with virtual load}\label{subsec:virtualload}

Finally, we provide an experimental validation of indirect shaping of the collective power consumption profile~\eqref{eqn:Total_Power} for DSM purposes. The method consists in introducing one agent, denoted as Virtual Power Load (VPL), which executes the TCL cooperation protocol with a pre-defined power consumption profile which is fictitious. Such  pre-defined power consumption profile consists in the amout of power demad that we wish to reduce in a given interval of time. This VPL cooperates with other agents only via dynamic consensus, i.e., task $a$ of Algorithm~\ref{mainalgo}. In this experiment we consider the case where during daily peak demand in an urban environment we wish to reduce to zero the power consumption of the network of TCLs for $5$ minutes.

To do so we design the VPL power consumption profile as large
in the interval $t\in[30\div35]\mathrm{min}$, in particular equal to  $100\mathrm{kW}$, and zero otherwise. The agent representing the VPL is connected to randomly and anonymously to other agents in the network.
The result of the test is shown in Figure~\ref{fig:vpl_test}. Since the VPL power consumption is greater than that of the network, by the online optimization of the proposed global objective function the emergent behavior of the TCL cooperation protocol shifts the power consumption of the real devices away from that of the VPL. We remark that this task is achieved without direct control of the load and preserving anonimity of the actual TCLs that shift their power consumption and while preserving the operating temperature range of each device, i.e., no disservice is caused to the users as a consequence of the DSM scheme.



\section{Conclusions} \label{sec:conclusion}

In this paper we presented a multi-agent DSM control architecture for the coordination of anonymous networks of thermostatically controlled loads via smart power sockets able to measure power consumption. We proposed: i) a method for parameter identification of a power consumption model; ii) an observer for the local estimation of the internal state of each TCL; iii) an MLD model of the TCL-plus-SPS hybrid system which enables numerical optimization of control inputs and prediction of power consumption; iv) A distributed and randomized online optimization method which enables the SPSs to cooperate autonomously and anonymously to optimize a constrained global objective function to enable demand side management;  v) a low cost experimental testbed based on off-the-shelf hardware an purpouse built software which we make available to the research community; vi) An experimental validation of the proposed method and architecture.


\bibliographystyle{IEEEtran}
\bibliography{autosam}

\end{document}